\documentclass[a4paper,UKenglish]{lipics-v2018}

\usepackage{microtype}%

\title{The Parikh Property for Weighted Context-Free Grammars}
\titlerunning{The Parikh Property for Weighted Context-Free Grammars} %

\author{Pierre Ganty}{IMDEA Software Institute, Madrid, Spain}{pierre.ganty@imdea.org}{}{}%

\author{Elena Gutiérrez}{IMDEA Software Institute, Madrid, Spain\\ Universidad Politécnica de Madrid, Spain}{elena.gutierrez@imdea.org}{}{}

\authorrunning{P. Ganty and E. Gutiérrez} %

\Copyright{Pierre Ganty and Elena Gutiérrez}%

\subjclass{}%

\EventEditors{}
\EventNoEds{0}
\EventLongTitle{38th IARCS Annual Conference on
Foundations of Software Technology and Theoretical Computer Science (FSTTCS 2018)}
\EventShortTitle{FSTTCS 2018}
\EventAcronym{FSTTCS}
\EventYear{2018}
\EventDate{December 10--14, 2018}
\EventLocation{Ahmedabad, Gujarat}
\EventLogo{}
\SeriesVolume{}
\ArticleNo{}
\acknowledgements{We thank Miguel Ambrona for pointing us to the theory of Groebner bases.}

\usepackage[utf8]{inputenc}
\usepackage{newunicodechar}
\newunicodechar{ε}{\varepsilon}
\newunicodechar{⟨}{\langle}
\newunicodechar{⟩}{\rangle}
\newunicodechar{⟦}{\llbracket}
\newunicodechar{⟧}{\rrbracket}
\newunicodechar{⌈}{\lceil}
\newunicodechar{⌉}{\rceil}
\newunicodechar{⇃}{\mathclose{\downharpoonleft}}
\newunicodechar{…}{\ldots}
\newunicodechar{⋯}{\cdots}
\newunicodechar{∈}{\in}
\newunicodechar{∉}{\notin}
\newunicodechar{⊕}{\oplus}
\newunicodechar{⊗}{\otimes}
\newunicodechar{∘}{\circ}
\newunicodechar{α}{\alpha}
\newunicodechar{φ}{\varphi}
\newunicodechar{ψ}{\psi}
\newunicodechar{μ}{\mu}
\newunicodechar{π}{\pi}
\newunicodechar{τ}{\tau}
\newunicodechar{Σ}{\Sigma}
\newunicodechar{≤}{\leq}
\newunicodechar{≥}{\geq}
\newunicodechar{≝}{\stackrel{\rm\scriptscriptstyle def}{=}}
\newunicodechar{𝒪}{\mathcal{O}}
\newunicodechar{𝒯}{\mathcal{T}}
\newunicodechar{𝒴}{\mathcal{Y}}
\newunicodechar{Δ}{\Delta}

\def\len#1{{\vert{#1}\vert}}

\newcommand{\dgrammar}{G^{\,⌈k⌉}}
\newcommand{\dvars}{V^{\,⌈k⌉}}
\newcommand{\drules}{R^{\,⌈k⌉}}
\newcommand{\dweights}{W^{\,⌈k⌉}}
\newtheorem*{theorem*}{Theorem}
\newtheorem*{lemma*}{Lemma}
	
\usepackage{amsmath}
\usepackage{amssymb}
\usepackage{stmaryrd}
\usepackage{forest}
\usepackage{mathtools}
\usepackage{paralist}
\usepackage{fancyvrb}

\VerbatimFootnotes

\hideLIPIcs
\nolinenumbers
\begin{document}

\maketitle

\begin{abstract}
Parikh's Theorem states that every context-free grammar (CFG) is equivalent to some regular CFG when the ordering of symbols in the words is ignored.
The same is not true for the so-called weighted CFGs, which additionally assign a weight to each grammar rule.
If the result holds for a given weighted CFG \(G\), we say that \(G\) satisfies the Parikh property.
We prove constructively that the Parikh property holds for every weighted nonexpansive CFG.
We also give a decision procedure for the property when the weights are over the rationals.
 \end{abstract}

\section{Introduction}
The celebrated Parikh's Theorem~\cite{Parikh1966} establishes that every context-free language is \emph{Parikh-equivalent} to some regular language.
Two words \(w, w'\) over an alphabet of symbols are Parikh-equivalent if the number of occurrences of each symbol in \(w\) coincides with that of \(w'\).
For instance, the words \(aabb\) and \(abab\) over the alphabet \(\{a,b\}\) are Parikh-equivalent as both have \(2\) \(a\)'s and \(2\) \(b\)'s.
Two languages \(L\) and \(L'\) are Parikh-equivalent if for each word in \(L\) there is a Parikh-equivalent word in \(L'\), and viceversa, e.g., the language \(\{ab, aabb\}\) is Parikh-equivalent to the language \(\{ba, abab, abba\}\).
Consider, for instance, the context-free language \(L = \{a^nb^n \mid n\geq 0\}\).
Then, a regular language that satisfies Parikh's Theorem is \((ab)^*\).
In fact, given a context-free grammar, one can construct a finite-state automaton that recognizes a Parikh-equivalent language \cite{Esparza2011}.
Parikh's Theorem has been applied in automata theory for decision problems concerning Parikh-equivalence such as membership, universality, equivalence and disjointness~\cite{Esparza1997,Huynh1985,Huynh1980,Huynh1982}, to establish complexity bounds on verification problems for counter machines~\cite{Goller2009}, equational Horn clauses \cite{Verma2005}, among many others.
It has also found application in the analysis of asynchronous programs with procedures~\cite{Ganty2010,Mahesh2006} where the Parikh-equivalent finite-state automaton is used to compute another asynchronous program without procedures that preserves safety bugs.

Weighted finite-state automata are a generalization of the classical nondeterministic finite-state automata in which each transition carries a weight.
This weight can be defined, for instance, as a nonnegative number representing the cost of its execution.
Then, the weight of a path in the weighted automaton can be computed by adding the weights of its transitions.
If we are interested in the minimal cost of execution of a given word, we can compute its weight as the minimum of the weights of the paths accepting that word.
In general, the algebraic structure underlying the computation of the weights is that of a \emph{semiring}, an algebraic structure with two operations \(\cdot\) (product) and \(+\) (sum) used to compute the weight of a path and the weight of a word, respectively.
In the same way, it is possible to add weights to the transitions of a pushdown automaton.
The later model, so-called weighted pushdown automata, has been used to perform data-flow analysis of programs with procedures~\cite{Reps2005}.

In this paper we study the question of whether Parikh's Theorem can be extended to the weighted case.
Roughly speaking, for a given weighted pushdown automaton \(\mathcal{P}\), we ask whether there is a weighted finite-state automaton \(\mathcal{F}\) that accepts a Parikh-equivalent language and such that for every word \(w\), the sum of the weights of all words Parikh-equivalent to \(w\) in \(\mathcal{P}\) coincides with that of all Parikh-equivalent words to \(w\) in \(\mathcal{F}\). 
Extending Parikh's Theorem to the weighted case has the potential of reaching new applications, for instance, the analysis of event-driven asynchronous programs with procedures where each transition is augmented with the probability of the event associated to it.
Finding a weighted finite-state automaton that is Parikh-equivalent to the original program and preserves the probabilities can be used to perform probabilistic analysis of programs following this paradigm.

We will present our results using the grammar model (as opposed to the automata model).
It is well-known that both models are equivalent, in the sense that both representations generate the same family of languages of weighted words.
Using weighted context-free grammars (WCFGs for short) allows us to exploit their connection with algebraic systems of equations to give more simple and convincing proofs of our results.
In a WCFG, a weight is assigned to each rule of the grammar.
The notion of weight is extended from rules to parse trees by multiplying the weights of the rules used along a tree, and from parse trees to words by adding the weights of all the possible parse trees that yield to a word.
We say that two WCFGs \(G_1\) and \(G_2\) are Parikh-equivalent if for each Parikh-equivalence class \(\mathcal{E}\), the sum of the weights of every word in \(\mathcal{E}\) under \(G_1\) and \(G_2\) coincide.

We consider the following problem: given a WCFG \(G\), does there exist a Parikh-equivalent WCFG \(G'\) that is regular?
If the answer is positive we say that \(G\) satisfies the Parikh property.
It follows from a known counterexample by Petre~\cite{Petre1999} that the property is not true in general.
Recently, Bhattiprolu et al.~\cite{Vijay17} further investigated this question.
They show a class of WCFGs over the unary alphabet that always satisfy the Parikh property.
Now, we show that every \emph{nonexpansive} WCFG (over an arbitrary alphabet and arbitrary semiring) satisfies the Parikh property.
A WCFG is nonexpansive if no grammar derivation is of the form \(X \Rightarrow^{*} w_0\,X\,w_1\,X\,w_2\).
Note that nonexpansiveness is decidable as it reduces to computing predecessors of a regular set~\cite{Esparza2000}.
We can show that in the unary case the class of nonexpansive grammars strictly contains the class defined by Bhattiprolu et al.~\cite{Vijay17} (see Appendix~\ref{sec:appendix-poly}).
However, nonexpansiveness is a sufficient condition for the Parikh property, but not necessary.
In particular, we give an example of an expansive WCFG for which there exists a Parikh-equivalent regular WCFG. %
This shows that a conjecture formulated by Baron and Kuich~\cite{Kuich1981} in 1981 is false\footnote{Essentially, they conjectured that every unambiguous WCFG \(G\) is nonexpansive iff \(G\) has the Parikh property \cite[Conjecture C]{Kuich1981}.}.
Furthermore, we can show that nonexpansiveness is not necessary for the property even when the alphabet is unary by means of a similar example.

In the second part of our work, we study the question of whether the Parikh property is decidable.
As far as we can tell, this question is open.
However, it implicitly follows from a result by Kuich et al.~\cite{KuichSalomaa1986} that, when we equivalently formulate the  property in terms of formal power series, it is decidable over the semiring of rational numbers.
Their proof relies on an ad-hoc elimination procedure which is hard to perform even on small examples.
Now we give a decision procedure that sidesteps this problem by using a different technique that allows to illustrate the algorithm on small examples with the support of mainstream open-source computer algebra systems.

The document is organized as follows.
After preliminaries in Section~\ref{sec:preliminaries}, we show in Section~\ref{sec:nonexpansiveness} that every nonexpansive WCFG is Parikh-equivalent to a regular WCFG.
In Section~\ref{sec:decidability}, we give a decision procedure for the property when the weight domain is over the rational numbers and we illustrate its use with several examples.
Finally, we give further details of the related work in Section~\ref{sec:relatedwork}, and conclusions and further work in Section~\ref{sec:conclusions}.
Missing proofs can be found in the Appendix.

\section{Preliminaries}
\label{sec:preliminaries}
We denote by \(Σ^{*}\) (\(Σ^{⊕}\)) the free (commutative) monoid generated by \(Σ\).
The elements of \(Σ^{*}\) are written as words over the alphabet \(Σ\), typically denoted by \(w, w'\) and \( w_i\) (\(i \in \mathbb{N}\)), while the elements of \(Σ^{⊕}\) are written as monomials in the variables \(Σ\) and they are typically denoted by \(v, v'\) and \(v_i\).
For instance, if \(Σ = \{a,b\}\) then all the elements in \(Σ^{*}\) of length two containing \(1\) \(a\) and \(1\) \(b\) are the words \(ab\) and \(ba\) while the only element with that property in \(Σ^{⊕}\) is the monomial \(ab\).

We denote a \emph{context-free grammar} (CFG for short) as a tuple \( (V,Σ,S,R)\) where \(V\) is a finite set of \emph{variables} including \(S\), the \emph{start} variable, \(Σ\) is the set of \emph{terminals} and \(R \subseteq V \times (Σ \cup V)^* \) is a finite set of \emph{rules}. Rules are conveniently denoted \(X \rightarrow \gamma\).
We will always assume that CFGs are \emph{cycle-free}, i.e., there is no derivation of the form \(X \Rightarrow^{+} X\) with \(X \in V\).
This guarantees that the number of parse trees for one given word is finite and thus the weight of a word is a well-defined function.
W.l.o.g., we assume that every \emph{regular} CFG is \emph{right-regular} , i.e., \(\gamma \in \Sigma^+(ε \cup V)\) for each \(\gamma\).
A CFG is \emph{nonexpansive} if no derivation is of the form \(X \Rightarrow^{*} w_0\,X\,w_1\,X\,w_2\) with \(X \in V\) and \(w_i \in (Σ \cup V)^*\).
Otherwise, it is \emph{expansive}.

A \emph{semiring} is a structure \((A,+,\cdot,0_{A},1_{A})\) where \((A,+,0_{A})\) is a commutative monoid with identity \(0_{A}\), \((A,\cdot,1_{A})\) is a monoid with identity \(1_{A}\), \(\cdot\) distributes over \(+\) and \(0_A\) satisfies that \(a \cdot 0_{A} = 0_{A} \cdot a = 0_{A}, \text{ for all } a \in A\).
A semiring is called \emph{commutative} iff \(a \cdot b = b \cdot a\) for every \(a, b \in A\).
In the sequel, we assume that \(A\) is always a commutative semiring.
An \emph{idempotent} semiring is one that satisfies \(a + a = a\), for all \(a \in A\), .
A (commutative)\emph{ ring} is a (commutative) semiring where \((A,+,0_{A})\) is a commutative \emph{group} (i.e., every element in \(A\) has an additive inverse).
Finally, a \emph{field} is a ring where \((A\setminus\{0_A\},\cdot,1_{A})\) is a \emph{commutative group} (i.e., every element in \(A\) except \(0_A\) has a multiplicative inverse).
We sometimes use \(A\) for both the structure and the underlying set when the meaning is clear from the context.
We abuse notation and use \(+\) and \(\cdot\) to denote the ordinary sum and product in \(\mathbb{N}\) and \(\mathbb{Q}\).
Classical examples of a commutative semirings are \((\mathbb{N}, +, \cdot, 0, 1)\) and \((\mathbb{Q}, +, \cdot, 0, 1)\).
The later is also a \emph{field} and we will refer to it as the \emph{rational semiring}.
Another classical example of a commutative semiring is the \emph{tropical semiring}, defined as \((\mathbb{N} \cup \{\infty\}, min, +, \infty, 0)\).
Note that this semiring is also idempotent as \(min(a,a) = a\), for all \(a \in \mathbb{N}\cup \{\infty\}\).

A \emph{weighted context-free grammar} (WCFG for short) is a pair \((G,W)\) where \(G\) is a CFG as defined above and \(W\) is a mapping with the signature \(W: R \rightarrow A\) that assigns a weight from \(A\) to each production in \(R\), for some (commutative) semiring \(A\).
Note that \(W\) may assign \(0_A\) to some rules in \(R\).
The mapping \(W\) is usually referred to as the \emph{weight function} of the WCFG.
We extend the definition of \(W\) from rules to \emph{derivation sequences}\footnote{For a definition of \emph{derivation sequence} go to the beginning of Appendix~\ref{sec:appendix-sec1}.} by assigning to each derivation sequence \(\psi\) a weight value which is the product of the weights of the rules applied in \(\psi\).
We assume that, the derivation policy for \(G\), i.e., the derivation strategy that determines the next variable to rewrite along a derivation, defines one unique derivation sequence for each parse tree.
We also assume that the \(\cdot\) operation is commutative, i.e., we will always consider \emph{commutative semirings}.
Then, the weight of a derivation sequence does not depend on the choice of the derivation policy. 
Under these assumptions we can extend the definition from rules to \emph{parse trees} (instead of derivation sequences).
Before, we recall some definitions.
We define a \emph{labeled tree} \(c(τ_1, \ldots, τ_n)\) (with \(n \geq 0\)) as a finite tree whose nodes are labeled, where \(c\) is the label of the root and \(τ_1, \ldots,τ_n\) are labeled trees, the children of the roots.
When \(n = 0\) we prefer to write \(c\) instead of \(c()\).
We simply write \(τ = c(\ldots)\) when the children nodes \(τ_1, \ldots, τ_n \) are not important.
We will write parse trees  as labeled trees of the form \(τ = \pi (τ_1, \ldots,τ_n)\) to denote that the topmost level of \(τ\) is induced by the grammar rule \(\pi\) and has exactly \(n\) children nodes which root (from left to right) the parse trees \(τ_1, \ldots, τ_n\), i.e., the right-hand side of \(\pi\) contains \(n\) grammar variables where the \(i\)-th (from the left) is derived according to \(τ_i\).
We thus define the \emph{yield} of a parse tree \(τ = \pi(τ_1, \ldots,τ_n)\), denoted as \(𝒴(τ)\) inductively as follows. If \(n = 0\), then \(𝒴(τ) = \gamma\) where \(\pi\) is of the form \(X \rightarrow \gamma\) and \(\gamma \in \Sigma^*\cup\{\varepsilon\}\).
Otherwise, \(𝒴(τ) = \alpha_1 𝒴(τ_1)\dots \alpha_n 𝒴(τ_n)\alpha_{n+1} \) where \(\pi\) is of the form \(X \rightarrow \alpha_1X_1\ldots \alpha_nX_n\alpha_{n+1}\) with \(\alpha_i \in \Sigma^*\cup\{\varepsilon\}\), and each \(X_i\) corresponds to the left-hand side of the rule in the root of \(τ_i\).
Define the \emph{weight of a parse tree} \(τ = \pi (τ_1, \ldots,τ_n) \) inductively as:
\[W(τ) ≝ W(\pi) \prod_{i=1}^{n} W(τ_i) \enspace .\]
Note that \(W(τ)\) does not depend on the order in which we consider the rules in \(τ\) as we assume that \(\cdot\) is commutative.
Denote by \(\mathcal{T}_G\) the set of all parse trees of a CFG \(G\).
Then, define the \emph{weight of a word} \(w \in Σ^{*}\) as follows:
\[W(w) ≝ \sum\limits_{\substack{𝒴(τ) = w\\τ \in𝒯_G }} W(τ) \enspace .\]
If for some \(w \in Σ^{*}\), the set \(\{τ \mid 𝒴(τ) = w, τ \in 𝒯_G \} = \emptyset\) then \(W(w) ≝ 0_A\).
For the following definitions we adopt a similar notation as Bhattiprolu et al.~\cite{Vijay17}.
Define the \emph{semantics} of a WCFG \((G,W)\), denoted by \(⟦G⟧_W\), as the mapping \(⟦G⟧_W: Σ^{*} \rightarrow A \) such that \(⟦G⟧_W(w)≝ W(w)\).
Define the \emph{Parikh image} of a word \(w \in \Sigma^*\) with \(\Sigma = \{a_1, \ldots, a_n\}\), denoted by \(\lbag w \rbag\) as the monomial \(a_1^{\alpha_1}a_2^{\alpha_2} \ldots a_n^{\alpha_n} \in Σ^{⊕} \) such that \(\alpha_i\) is the number of occurrences of \(a_i\) in \(w\). 
Define the \emph{Parikh image} of a weighted context-free grammar \((G,W)\), denoted by \(Pk⟦G⟧_W\), as the mapping \(Pk⟦G⟧_W: Σ^{⊕} \rightarrow A \)  such that:
\[Pk⟦G⟧_W(v) ≝ \sum\limits_{\substack{v = \lbag w \rbag \\ w \in Σ^{*} } }⟦G⟧_W(w) \enspace . \]
We write \(⟦G⟧_W\) and \(Pk⟦G⟧_W\) as the formal sums \(\sum_{w \in Σ^{*}} ⟦G⟧_W(w)\,w\) and \(\sum_{v \in Σ^{⊕}} Pk⟦G⟧_W(v)\,v\), respectively.
Two WCFGs \((G,W)\) and \((G',W')\) are \emph{language-equivalent} iff \(⟦G⟧_W = ⟦G'⟧_{W'}\), while \((G,W)\) and \((G',W')\) are \emph{Parikh-equivalent} iff \(Pk⟦G⟧_W = Pk⟦G'⟧_{W'}\).
Finally, a WCFG \((G,W)\) is \emph{regular}/\emph{nonexpansive}/\emph{cycle-free} iff \(G\) is regular/nonexpansive/cycle-free, respectively.

\begin{definition}[Parikh property]
A WCFG \((G,W)\) satisfies the \emph{Parikh property} iff there exists a WCFG \((G_{\ell}, W_{\ell})\) such that:
\begin{enumerate}
\item \((G_{\ell}, W_{\ell})\) is regular, and
\item\(Pk⟦G⟧_W = Pk⟦G_{\ell}⟧_{W_{\ell}}\).
\end{enumerate}
\end{definition}

\section{Sufficient condition for the Parikh property}
\label{sec:nonexpansiveness}
Petre~\cite{Petre1999} shows that the Parikh property is not true in general.
In the following example we show a well-known WCFG (for instance, see \cite{Vijay17,Petre1999}) for which no regular Parikh-equivalent WCFG exists.

\begin{example}
\label{ex:example0}
Consider the WCFG \((G,W)\) with \(G = (\{X\}, \{a\}, X, \{X \rightarrow aXX,\,X \rightarrow a\})\) and the weight function \(W\) over \((\mathbb{N}, +, \cdot, 0, 1)\) that assigns \(1\) to each production in the grammar.
Note that, because the alphabet is unary, we have that \(Pk⟦G⟧_W = ⟦G⟧_W\).
As \(W\) assigns \(1\) to each grammar rule, the weight of each word can be interpreted as its ambiguity according to \(G\).
Then, the reader can check that: 
\[⟦G⟧_W = \sum_{n\geq 0} C_n\,a^{2n+1} = 1a + 1a^3 + 2a^5 + 5a^7 + 14a^9 + 42a^{11} + 132a^{13} + 429a^{15}+\ldots \]
with \(C_n = \frac{1}{n+1}\binom{2n}{n}\) the \(n\)-th Catalan number.
We will see in Example~\ref{ex:example3} that this formal power series cannot be generated by a regular WCFG.
\end{example}

Now we show that every nonexpansive WCFG over an arbitrary commutative weight domain satisfies the Parikh property.

\begin{theorem}
\label{thm:nexp-parikh}
Let \((G,W)\) be an arbitrary WCFG.  If \(G\) is nonexpansive then \((G,W)\) satisfies the Parikh property.
\end{theorem}

\begin{proof}
The proof is constructive.
Here we give the main intuition of the construction.
For a complete proof go to Appendix~\ref{sec:appendix-sec1}.
For every nonexpansive WCFG \((G,W)\), we give a 2-step construction that results in a Parikh-equivalent regular WCFG \((G_{\ell}, W_{\ell})\).
The steps are:
\begin{enumerate}
\item construct a new WCFG \(\big(\dgrammar, \dweights\big)\), with \(k\in \mathbb{N}\), language-equivalent to \((G,W)\); and
\item construct a regular WCFG \((G_{\ell}, W_{\ell})\) Parikh-equivalent to \(\big(\dgrammar, \dweights\big)\). 
\end{enumerate}

The general idea behind the first step is to build a WCFG \(\big(\dgrammar, \dweights\big)\) that contains all the information needed to define a ``strategic'' derivation policy.
This derivation policy is strategic in the sense that the total number of grammar variables in all \emph{derivation sentences}\footnote{For a definition of \emph{derivation sentence} go to the beginning of Appendix~\ref{sec:appendix-sec1}.} produced along a derivation sequence is bounded.
To build \(\big(\dgrammar, \dweights\big)\) we rely on the grammar construction given by Luttenberger et al.~\cite{Luttenberger2016}.

In the second step of the construction, we use \(\big(\dgrammar, \dweights\big)\) to build a regular WCFG \((G_{\ell}, W_{\ell})\) that is Parikh-equivalent.
Each grammar variable of \((G_{\ell}, W_{\ell})\) represents each possible sentence (without the terminals) along a derivation sequence of \(\big(\dgrammar, \dweights\big)\), and each rule simulates a derivation step of \(\big(\dgrammar, \dweights\big)\).
Because the number of variables in the sentences is bounded, the number of variables and rules of \((G_{\ell}, W_{\ell})\) is necessarily finite.
This construction is very similar to that given by Bhattiprolu et al.~\cite{Vijay17} and Esparza et al.~\cite{Esparza2011}.

\end{proof}

The converse of Theorem~\ref{thm:nexp-parikh} is not true.
The next counterexample illustrates this fact by defining an expansive WCFG \(G_2\)  for which a Parikh-equivalent regular WCFG \(G_1\) exists.
Thus, nonexpansiveness does not provide an exact characterization of the Parikh property.

\begin{example}
\label{ex:counterexample2}
Consider the WCFG \((G_1,W_1)\) where \(G_1 = (\{X_1\}, \{a,\overline{a}\},X_1, R_1 = \{X_1 \rightarrow aX_1,\,X_1 \rightarrow \overline{a}X_1,\,X_1 \rightarrow \varepsilon\})\) and \(W_1\) is defined over \((\mathbb{N}, +, \cdot, 0,1)\) and assigns \(1\) to each rule in \(R_1\).
First, note that \((G_1, W_1)\) is regular and the weight of each word can be interpreted as its ambiguity according to \(G_1\).
Because \(G_1\) is unambiguous, the weight of each word in the language of \(G_1\) is \(1\).
It is easy to see that \(⟦G_1⟧_{W_1}\) is:
\[⟦G_1⟧_{W_1} = (a + \overline{a})^* = \sum\limits_{n\geq0}(a + \overline{a})^n = 1 \varepsilon + 1a + 1\overline{a} + 1a\overline{a} + 1\overline{a}a + 1aaa + 1 aa\overline{a} + 1a\overline{a}{a} + 1a\overline{a}\overline{a}+\ldots\]

Now consider the expansive WCFG \((G_D,W_D)\) where \(G_D = (\{D\}, \{a,\overline{a}\},D,R_D = \{D \rightarrow aD\overline{a}D,\,D  \rightarrow \varepsilon\})\) and \(W_D\) is defined over \(\mathbb{N}\) and assigns \(1\) to each rule in \(R_D\).
The grammar \(G_D\) generates the Dyck language \(L_D\) over the alphabet \(\{a, \overline{a}\}\) and it is also unambiguous (i.e., the weight of each \(w \in L_D\) is \(1\)).
It is well-known that \(L_D\) is a deterministic context-free language (DCFL).
Then the complement of \(L_D\), namely \(\{a, \overline{a}\}^{*} \setminus L_D\), is also a DCFL and thus admits an unambiguous CFG.
Let \(G_{\overline{D}} = (V_{\overline{D}}, \{a, \overline{a}\}, \overline{D}, R_{\overline{D}})\) be the unambiguous CFG that generates \(\{a, \overline{a}\}^{*} \setminus L_D\), and define \((G_{\overline{D}}, W_{\overline{D}})\) where \(W_{\overline{D}}\) is defined over \(\mathbb{N}\) and assigns \(1\) to each rule in \(R_{\overline{D}}\).

W.l.o.g., assume \(V_D \cap V_{\overline{D}} = \emptyset\) and consider a new variable \(X_2\notin V_D \cup V_{\overline{D}}\).
Define the WCFG \((G_2, W_2)\) where \(G_2 = (\{X_2\} \cup V_D \cup V_{\overline{D}}, \{a, \overline{a}\}, X_2, R_2)\), \(R_2\) is defined as \(R_2 = \{X_2 \rightarrow D,\,X_2 \rightarrow \overline{D}\} \cup R_D \cup R_{\overline{D}}\) where  \(W_2\) is defined over \(\mathbb{N}\) and assigns \(1\) to each rule in \(R_2\).
First, \(G_2\) is expansive because \(G_D\) is expansive.
Furthermore, \(D\) and \(\overline{D}\) generate unambiguously languages that are complementary over \(\{a, \overline{a}\}\).
As the weight of each word in \((G_2,W_2)\) corresponds to its ambiguity, we have that \(⟦G_2⟧_{W_2} = (a + \overline{a})^{*}\).
Hence \(⟦G_1⟧_{W_1} =⟦G_2⟧_{W_2}\) and thus \(Pk⟦G_1⟧_{W_1} = Pk⟦G_2⟧_{W_2}\).
Recall that \((G_1, W_1)\) is regular.
We conclude that \((G_2,W_2)\) is expansive and satisfies the Parikh property.
\qed
\end{example}
We can give a similar counterexample over a unary alphabet (see Appendix~\ref{sec:appendix-counterexample}).
This shows that nonexpansiveness is not necessary for the Parikh property even in the unary case.

\section{A decision procedure for the Parikh property over the rationals}
\label{sec:decidability}

In this section we give a decision procedure that tells whether or not a given WCFG with weights over the rational semiring satisfies the Parikh property.
Our procedure relies on a decidability result by Kuich and Salomaa ~\cite[Theorem 16.13]{KuichSalomaa1986}.
It implicitly follows from this result that the Parikh property is decidable over the rational semiring.
However, their decision procedure is hard to follow as it relies on algebraic methods beyond the scope of this field.
This makes its implementation rather involved even for small instances.
We propose an alternative method to sidestep this problem using Groebner basis theory.

First, we give some preliminaries.
In what follows, \(A\) will denote a partially ordered commutative semiring.
Given \(A\) and an alphabet \(Σ\), a \emph{formal power series in commuting variables} is a mapping of \(Σ^{⊕}\) into \(A\).
\(A⟨⟨Σ^{⊕}⟩⟩\) denotes the set of all formal power series in commuting variables \(Σ\) and coefficients in \(A\).
The values of a formal power series \(r\) are denoted by \((r,v)\) where \(v \in Σ^{⊕}\).
As \(r\) is a mapping of \(Σ^{⊕}\) into \(A\), it can be written as a formal sum as \(r = \sum_{v \in Σ^{⊕}} (r,v)\,v\).
When \(v = ε\) we will write the term \((r,ε)ε\) of \(r\) simply as \((r,ε)\).
We define the \emph{support} of a formal power series as \(supp(r) ≝\{v \mid (r,v) \neq 0_A\}\).
The subset of \(A⟨⟨Σ^{⊕}⟩⟩\) consisting of all series with a finite support is denoted by \(A⟨Σ^{⊕}⟩\) and its elements are called \emph{polynomials}.
Finally, define, for \(k \geq 0\), the operator \(R_k\) by \(R_k(r) ≝ \sum\limits_{|v| \leq k}(r,v)v \) where \(r \in A⟨⟨Σ^{⊕}⟩⟩\).

Now we establish the connection between WCFGs and algebraic systems in commuting variables.
Let \((G,W)\) be a WCFG with \(G = (V, Σ, X_1, R)\), \(V = \{X_1, \ldots,X_n\}\), and \(W\) defined over the semiring \(A\).
We associate to \((G,W)\) the algebraic system in commuting variables defined as follows.
For each \(X_i \in V\):
\begin{equation}
\label{eq:algebraic-system}
X_i = \sum\limits_{\substack{\pi \in R \\ \pi = (X_i \rightarrow \gamma)}} W(\pi) \lbag\gamma\rbag \enspace .
\end{equation}
We refer to this system as the \emph{algebraic system (in commuting variables) corresponding to \((G,W)\)}.
Sometimes, we write \(A⟨Σ^{⊕}⟩\)-algebraic system to indicate that the coefficients of the system lie in \(A⟨Σ^{⊕}⟩\).
Note that \eqref{eq:algebraic-system} can be written as follows.
For each \(X_i \in V\):
\begin{equation}
\label{eq:algebraic-system-poly}
X_i = p_i \enspace , \text{ with \(p_i \in A⟨(Σ \cup V)^{⊕}⟩\)} \enspace .
\end{equation}
A \emph{solution} to \eqref{eq:algebraic-system-poly} is defined as an \(n\)-tuple \(r = (r_1, \ldots, r_{n})\) of elements of \(A⟨⟨Σ^{⊕}⟩⟩\) such that \(r_i = r(p_i)\), for \(i = 1, \ldots, n\), where \(r(p_i)\) denotes the series obtained from \(p_i\) by replacing, for \(j = 1, \ldots, n\), simultaneously each occurrence of \(X_j\) by \(r_j\).
Note that, \(r_1\), the first component of \(r\), always corresponds to the solution for \(X_1\), the initial variable of \(G\).
The \emph{approximation sequence} \(\sigma^0, \sigma^1, \ldots, \sigma^j, \ldots\) where each \(\sigma^j\) is an \(n\)-tuple of elements of \( A⟨Σ ^{⊕}⟩\) associated to an algebraic system as \eqref{eq:algebraic-system-poly} is defined as \(\sigma^0 = (0_A, \ldots, 0_A)\) and \(\sigma^{j+1} = (\sigma^{j}(p_1), \ldots, \sigma^{j}(p_n))\) for all \(j\geq 0\).
We have that \(\lim_{j\to\infty} \sigma^j = \sigma\) iff for all \(k\geq 0\) there exists an \(m(k)\) such that \(R_k(\sigma^{m(k)+j}) = R_k(\sigma^{m(k)}) = R_k(\sigma)\) for all \(j\geq0\).
If \(\lim_{j\to\infty} \sigma^j = \sigma\), then \(\sigma\) is a solution of \eqref{eq:algebraic-system-poly} (from Theorem 14.1 in \cite{KuichSalomaa1986}) and is referred to as the \emph{strong solution}.
Note that, by definition, the strong solution is unique whenever it exits.
Finally, if \((G,W)\) is a regular WCFG then each \(p_i\) in its corresponding algebraic system written as in \eqref{eq:algebraic-system-poly} is a polynomial in \(A⟨\mathcal{M}⟩\), where \(\mathcal{M}\) denotes the set of monomials of the form \(a_1^{\alpha_1} \ldots a_m^{\alpha_m}\,X_1^{\beta_1}\ldots X_n^{\beta_n}\) with \(a_i \in \Sigma\), \(\alpha_i,\beta_j \in \mathbb{N}\) for all \(i\) and \(j\), and \( \sum_{i=1}^{n}\beta_i \leq 1\).
We call a system of this form a \emph{regular algebraic system}.
Conversely, we associate to each \(A⟨Σ^{⊕}⟩\)-algebraic system \(S\) in commuting variables of the form \eqref{eq:algebraic-system-poly} a WCFG \((G,W)\) over the semiring \(A\) as follows.
Define \(G = (\{X_1, \ldots, X_n\}, Σ, X_1, R)\) and such that \(\pi = (X_i \rightarrow \gamma) \in R\) iff \((p_i,\gamma) \neq 0_A\).
If \(\pi \in R\) then \(W(\pi) = (p_i,\gamma)\).
We will refer to \((G,W)\) as the \emph{WCFG corresponding to the algebraic system \(S\)}.
Note that if we begin with an algebraic system in commuting variables, then go to the corresponding WCFG and back again to an algebraic system, then the latter coincides with the original.
However, if we begin with the WCFG, form the corresponding algebraic system and then again the corresponding WCFG, then the latter grammar may differ from the original.

Next theorem shows that the Parikh image of a cycle-free WCFG corresponds to the solution for the initial variable in the corresponding algebraic system.
\begin{theorem}
\label{thm:Parikh-image}
Let \((G,W)\) be a cycle-free WCFG and let \(S\) be the algebraic system in commuting variables corresponding to \((G,W)\).
Then, the strong solution \(r\) of \(S\) exists and the first component of \(r\) corresponds to \(Pk⟦G⟧_W\).
\end{theorem}

Now we introduce the class of \emph{rational} power series in commuting variables \(Σ\) with coefficients in the semiring \(A\), denoted by \(A^{rat}⟨⟨Σ^{⊕}⟩⟩\).
\begin{definition}
\label{thm:rational-series}
\(r \in A^{rat}⟨⟨Σ^{⊕}⟩⟩\) iff \(r\) is the first component of the solution of a regular algebraic system in commuting variables.
\end{definition}

From the previous definition and Theorem~\ref{thm:Parikh-image} we can characterize the WCFGs that satisfy the Parikh property as follows.
\begin{lemma}
\label{lemma:regular-prop}
Let \((G,W)\) be a cycle-free WCFG.
Then \((G,W)\) satisfies the Parikh property iff \(Pk⟦G⟧_W \in A^{rat}⟨⟨Σ^{⊕}⟩⟩\).
\end{lemma}

Next we observe that every WCFG \((G,W)\) defined over a commutative \emph{ring} with the Parikh property satisfies a linear equation of a special kind.
This result directly follows from Theorem  16.4 in~\cite{KuichSalomaa1986}.
\begin{theorem}
\label{thm:linear-recurr}
Let \((G,W)\) be a cycle-free WCFG with \(W\) defined over a commutative ring \(A\). Then \((G,W)\) satisfies the Parikh property iff \(Pk⟦G⟧_W\) satisfies a linear equation of the form: \mbox{\(X = sX + t\)}, for some \(s,t \in A⟨Σ^{⊕}⟩\) with \((s,\varepsilon) = 0\).
\end{theorem}
\begin{proof}
The result is a consequence of Theorem  16.4 in~\cite{KuichSalomaa1986} and Lemma~\ref{lemma:regular-prop}.
\end{proof}

We conclude from the previous theorem that, given a WCFG \((G,W)\) with \(W\) defined over a commutative ring, if such a linear equation exists then \((G,W)\) satisfies the Parikh property; otherwise it does not.
Now we will use a result by Kuich et al.~\cite{KuichSalomaa1986}  to conclude that, if \((G,W)\) is defined over \(\mathbb{Q}\) then there exists an irreducible polynomial \(q(X)\) such that \(q\) evaluates to \(0\) when \(X = Pk⟦G⟧_W\), denoted by \(q(Pk⟦G⟧_W) \equiv 0\).
Intuitively, this polynomial contains all the information needed to decide whether or not \((G,W)\) has the Parikh property.

\begin{theorem}[from Theorem 16.9 in~\cite{KuichSalomaa1986}]
\label{thm:q-existence}
Let \(S\) be the \(\mathbb{Q}⟨Σ^{⊕}⟩\)-algebraic system in commuting variables corresponding to a cycle-free WCFG.
Let \(r_1\) be the first component of its strong solution.
Then there exists an irreducible polynomial \(q(X_1)\) with coefficients in \(\mathbb{Q}⟨Σ^{⊕}⟩\), and unique up to a factor in \(\mathbb{Q}⟨Σ^{⊕}⟩\), such that \(q(r_1) \equiv 0\).
\end{theorem}
Kuich et al.~\cite{KuichSalomaa1986} show that the polynomial \(q\) is effectively computable by means of a procedure based on the classical elimination theory.
Now we develop an alternative method using Groebner bases.
Before introducing this technique, we give some intuition on the ideas presented above by revisiting the examples of the previous section.

\begin{example}
\label{ex:example3}
Consider the cycle-free WCFG \((G,W)\) defined in Example~\ref{ex:example0} where the weight function \(W\) is now defined over \((\mathbb{Q}, +, \cdot, 0, 1)\) and assigns \(1\) to each production in the grammar.
The algebraic system \(S\) corresponding to \((G,W)\) is given by the equation \(X = a\,X^2 + a\).
Let \(r_1\) be its strong solution.
Assume for now that the irreducible polynomial \(q(X) \in \mathbb{Q}⟨\{a\}^{⊕}⟩⟨X⟩\) from Theorem~\ref{thm:q-existence} is \(q(X) = aX^2 - X + a\) (later we will give its construction using Groebner bases).
We will see later that the fact that \(q(X)\) is not linear is enough to conclude that \((G,W)\) does not satisfy the Parikh property (as we expected).
Note that the solution of \(S\) is \(r_1 = \frac{1 - \sqrt{1-4a^2}}{2a}\), which written as a series corresponds to \(r_1 = \sum_{n\geq 0} C_n\,a^{2n+1}\), with \(C_n = \frac{1}{n+1}\binom{2n}{n}\) the \(n\)-th Catalan number.
It is known that this formal power series cannot be written as the solution of a linear equation with coefficients in \(\mathbb{Q}⟨\{a\}^{⊕}⟩\)~\cite{Vijay17}.
\end{example}

\begin{example}
\label{ex:example2}
Now we will consider the WCFG given in Example~\ref{ex:counterexample2}.
This time we will give a complete definition of its grammar rules and, as in the previous example, we will extend its weight domain from \(\mathbb{N}\) to \(\mathbb{Q}\).
Define the WCFG \((G_2,W_2)\) where \(G_2 = (\{X_2,\overline{D},D,Y,Z\}, \{a,\overline{a}\}, X_2, R_2)\), \(R_2\) is given by:
\begin{align*}
X_2 &\rightarrow D \mid \overline{D} & \overline{D} &\rightarrow D\,\overline{a}\,Y \mid D\,a\,Z & Z &\rightarrow D\,a\,Z \mid D \enspace .
\\
D &\rightarrow a\,D\,\overline{a}\,D \mid \varepsilon & Y &\rightarrow a\,Y \mid \overline{a}\,Y \mid \varepsilon
\end{align*}
and the weight function \(W_2\) is defined over \((\mathbb{Q}, +, \cdot, 0, 1)\) and assigns \(1\) to each production in the grammar.
Note that \((G_2,W_2)\) is cycle-free.
The grammar variable \(D\) generates all the words in the Dyck language \(L_D\) over the alphabet \(\{a,\overline{a}\}\), while the variable \(\overline{D}\) generates \(\{a, \overline{a}\}^* \setminus L_D\).
The system \(S\) corresponding to \((G_2,W_2)\) consists of the following equations:
\begin{align*}
X_2 &= D + \overline{D} & \overline{D} &= D\,\overline{a}\,Y + D\,a\,Z & Z &= D\,a\,Z + D \enspace .
\\
D &= a\,D\,\overline{a}\,D + 1 & Y &= a\,Y + \overline{a}\,Y + 1
\end{align*}
Let \(\sigma = (r_1, r_2, r_3, r_4, r_5)\) be its strong solution where \(r_1\) corresponds to the solution for the initial variable \(X_2\).
Assume for now that the irreducible polynomial \(q(X_2) \in \mathbb{Q}⟨\{a,\overline{a}\}^{⊕}⟩⟨X_2⟩\) described by Theorem~\ref{thm:q-existence} is:
\[q(X_2) = (1 - (a + \overline{a}))X_2 - 1 \enspace .\]
We observe that \(q\) is linear in \(X_2\) and can be written as:
\[q(X_2) = (1 - s)X_2 - t = (1 - (a + \overline{a}))X_2 - 1 \enspace ,\]
with \((s, ε) =  0\).
Thus, by Theorem~\ref{thm:linear-recurr}, we conclude that \((G_2,W_2)\) satisfies the Parikh property as we expected.
\end{example}

Now we develop the technique we will use to construct the irreducible polynomial of Theorem~\ref{thm:q-existence}: Groebner bases.
A Groebner basis is a set of polynomials in one or more variables enjoying certain properties.
Given a set of polynomials \(F\) with coefficients in a field, one can compute a Groebner basis \(G\) of \(F\) with the property that \(G\) has the same solutions as \(F\) when interpreted as a polynomial system of equations.
Then, problems such as finding the solutions for the system induced by \(F\), or looking for alternative representations of polynomials in terms of other polynomials become easier using \(G\) instead of \(F\).
One of the main insights for using Groebner bases is that they are \emph{effectively constructable} using standard computer algebra systems, for any set of polynomials with coefficients in a field.

We are interested in computing Groebner bases of algebraic systems in commuting variables corresponding to weighted CFGs.
Given a WCFG and its corresponding algebraic system, our goal is to obtain a system with the same solution as the original, and such that one equation in the new system depends only on the initial grammar variable \(X_1\).
This equation will contain all the information needed to decide whether or not the given WCFG satisfies the Parikh property.
We will not enter into the technical details of how Groebner bases are constructed and their properties as these lie beyond the scope of this document (however, an explicit reference will be given in connection with each result applied).
Instead, we will give a result that encapsulates all the preconditions and postconditions we need for our purpose (Theorem~\ref{thm:g-properties}).
We first introduce the definitions that will appear in the theorem.

In what follows, \(K\) will always denote a field.
First we need to introduce the notion of \emph{ideal}.
Let \(K⟨V^{⊕}⟩\) denote the ring of polynomials in variables \(V\) and with coefficients in \(K\).
A subset \(I \subset K⟨V^{⊕}⟩ \) is an \emph{ideal} iff (i) \(0_K \in I\), (ii) if \(f, g \in I\) then \(f + g \in I\), and (iii) if \(f \in I\) and \(h \in K⟨V^{⊕}⟩\) then \(h \cdot f \in I\).
Given a set of polynomials \(F = \{f_1, \ldots, f_n\}\), we define \(⟨F⟩\) as \(⟨F⟩ ≝ \{\sum_{i=1}^{n} h_i\cdot f_i \mid h_i \in K⟨V^{⊕}⟩, f_i \in F\}\).
It can be shown that \(⟨F⟩\) is an ideal~\cite{Cox1997} and we call it the \emph{ideal generated by \(F\)}.
When an ideal is generated by a finite number of polynomials \(g_1, \ldots, g_n \in K⟨V^{⊕}⟩\), we say that \(g_1, \ldots, g_n\) is a \emph{basis} of the ideal.
It is known that every ideal in \(K⟨V^{⊕}⟩\) has a basis (actually many, but the ones we are particularly interested in are the so-called Groebner bases)~\cite{Cox1997}.
If one considers the set of polynomial equations \(\{f = 0 \mid f\in F\}\), denoted by \(F = \mathbf{0}\), then the set of all solutions of \(F = \mathbf{0}\) is defined as \(\{(r_1,r_2, \ldots, r_n)\in K^n \mid f(r_1, \ldots,r_n) \equiv 0, \text{ for all }f \in F\}\).
Then, given two sets of polynomials \(F\) and \(G\), if \(⟨F⟩ = ⟨G⟩\) then the set of solutions of \(F = \mathbf{0}\) coincides with the set of solutions of \(G = \mathbf{0}\)~\cite{Cox1997}.
To construct a Groebner basis of an ideal \(I\), one needs to impose first a total ordering on the monomials of variables occuring in \(I\).
This choice is significant as different orderings lead to different Groebner bases with different properties.
We are interested in computing Groebner bases with the \emph{elimination property} for the initial variable \(X_1\), i.e., bases where at least one polynomial depends only on \(X_1\).
Hence, we will always impose the \emph{ reverse lexicographic ordering} to construct Groebner bases.
\begin{definition}
\label{def:revlex}
Let \(V = \{X_1, \ldots,X_n\}\) be a set of variables.
Let \(\alpha\) and \(\beta\) be two monomials in \( V^{⊕}\) and let \(\overline{\alpha}\) (resp. \( \overline{\beta})\) be the vector in \(\mathbb{N}^n\) such that its \(i\)-th component corresponds to the number of occurrences of the variable \(X_i\) in \(\alpha\) (resp. \(\beta\)).
Then we say that \emph{\(\alpha\) is greater than \(\beta\) w.r.t. the reverse lexicographic ordering}, denoted by \(\alpha \succ_{revlex} \beta\), iff the first non-zero component of the vector \(\overline{\alpha} - \overline{\beta}\) is negative.
\end{definition}

Notice that Definition~\ref{def:revlex} implies an ordering of the variables: \(X_n \succ_{revlex} X_{n-1} \succ_{revlex} \ldots \succ_{revlex} X_1\).
The reason for choosing the \emph{reverse} lexicographic ordering is that, in order to compute a Groebner basis with the elimination property for the initial  variable \(X_1\), we need \(X_1\) to be the least monomial (with one or more variable).
In what follows, the phrase ``w.r.t. the reverse lexicographic ordering'' (for some given \(V = \{X_1, \ldots, X_n\}\)) will refer to the one described in Definition~\ref{def:revlex} with variables \(V\), unless stated otherwise. 
Fixed a total monomial ordering, we define the \emph{leading monomial} of a polynomial \(p\) as the greatest monomial in \(p\), and we denote it by \(LM(p)\).
We define the \emph{leading term} of \(p\) as the leading monomial of \(p\) together with its coefficient, and we denote it by \(LT(p)\).
Finally, we introduce the notion of a \emph{reduced} Groebner basis which allows to define uniquely a Groebner basis of an ideal of polynomials.
Let \(F\) be a set of polynomials and \(G\)  a Groebner basis of \(⟨F⟩\).
We say that \(G\) is a \emph{reduced} Groebner basis of \(⟨F⟩\) iff for each \(g_i \in G\) (i) the coefficient of \(LT(g_i) = 1\); and (ii) \(LM(g_i)\) does not divide any term of any \(g_j\) with \(i\neq j\).
For a given set of polynomials \(F\) and monomial ordering \(\succ\), there exists exactly one reduced Groebner basis of \(⟨F⟩\) w.r.t. \(\succ\)~\cite{Cox1997}.
We abuse notation and write \(K⟨X⟩\) instead of \(K⟨\{X\}^{⊕}⟩\) to refer to the ring of polynomials in the variable \(X\) with coefficients in \(K\).
Now we are ready to give the theorem.
\begin{theorem}
\label{thm:g-properties}
Let \(K\) be a field and \(V = \{X_1, \ldots,X_n\}\) a set of variables.
Let \(F \subseteq K⟨V^{⊕}⟩\) be a set of polynomials such that the strong solution of the system \(F = \mathbf{0}\) is \((r_1, \ldots,r_n)\) where \(r_i\) corresponds to the solution for \(X_i\).
Let \(G\) be the reduced Groebner basis of \(⟨F⟩\) w.r.t. the reverse lexicographic ordering.
Then the following properties are satisfied:
\begin{enumerate}
	\item \((r_1, \ldots, r_n\)) is also the strong solution of the system \(G = \mathbf{0}\) and,
	\item there is exactly one polynomial \(g \in G\) s.t. \(g\in K⟨X_1⟩\), and for that \(g\) we have \(g(r_1) \equiv 0\).
\end{enumerate}
\end{theorem}

\begin{proof}
Property \(1.\) follows from the fact that \(G\) is a basis of \(⟨F⟩\).
Now we prove property \(2.\)
\(G\) is a Groebner basis of \(⟨F⟩\) w.r.t. the reverse lexicographic ordering.
Then, as a result of the Elimination Theorem~\cite[Theorem 3.1.2]{Cox1997}, \(G \cap K⟨X_1⟩\) is a Groebner basis of \(⟨F⟩ \cap K⟨X_1⟩\).
Assume first that \(G \cap K⟨X_1⟩\) contains only the zero polynomial (the constant polynomial whose coefficients are equal to \(0\)).
Then the ideal \(⟨F⟩ \cap K⟨X_1⟩\) also contains only the zero polynomial.
But this contradicts Theorem~\ref{thm:q-existence}.
Then \(G \cap K⟨X_1⟩\) contains at least one nonzero polynomial \(g\).
Assume now that \(G \cap K⟨X_1⟩\) contains two different elements \(g_1\) and \(g_2\) in \(K⟨X_1⟩\).
W.l.o.g., let \(g_1\) be such that \(LM(g_1) \preceq_{lex} LM(g_2)\).
Thus, \(LM(g_1)\) divides (at least) the leading term of \(g_2\).
Then \(G\) is not in reduced form (contradiction).
We conclude that there is exactly one (nonzero) polynomial \(g \in G\) such that \(g\in K⟨X_1⟩\).
Finally, \(g(r_1)\equiv 0\) follows from \(1.\) and the fact that \( g \in (G \cap K⟨X_1⟩)\).
\end{proof}

Now we show in Theorem~\ref{thm:q-construction} how to construct \(q\) using Groebner bases.
Finally, we give in Theorem~\ref{thm:decision-proc} the main result of this section.

\begin{theorem}
\label{thm:q-construction}
Let \(S\) be a \(\mathbb{Q}⟨Σ^{⊕}⟩\)-algebraic system in commuting variables corresponding to a cycle-free WCFG and \(r_1\) be the first component of its strong solution.
Then an irreducible polynomial \(q(X_1)\) with coefficients in \(\mathbb{Q}⟨Σ^{⊕}⟩\) such that \(q(r_1) \equiv 0\) can be effectively constructed.
\end{theorem}
\begin{proof}

We begin with the first part of the algorithm.
Let \(K\) be the fraction field of \(\mathbb{Q}⟨Σ^{⊕}⟩\), i.e., the smallest field (w.r.t. inclusion) containing \(\mathbb{Q}⟨Σ^{⊕}⟩\). 
Consider \(S\) as defined in \eqref{eq:algebraic-system-poly} (page~\pageref{eq:algebraic-system-poly}) where now each polynomial \(p_i\) has its coefficients in \(K\) and its variables in \(V\), and let \(F\subseteq K⟨V^{⊕}⟩\) be the set of polynomials \(\{p_i \mid 1 \leq i \leq n\}\).
Construct the reduced Groebner basis \(G\) of \(F\) w.r.t. the reverse lexicographic ordering.
Let \(G = \{g_1, \ldots,g_s\}\) with \(s \geq 1\).
By Theorem~\ref{thm:g-properties}, there is exactly one \(g \in G\) such that \(g \in K⟨X_1⟩\), and \(g\) satisfies \(g(r_1) \equiv 0\).

We cannot conclude yet that \(g(X_1)\) is the polynomial \(q(X_1)\) we are looking for since \(g(X_1)\) might not be irreducible in the fraction field of \(\mathbb{Q}⟨Σ^{⊕}⟩\).
This constitutes the second part of the algorithm which follows the method given in~\cite{KuichSalomaa1986} to obtain from \(g(X_1)\) an irreducible polynomial \(q(X_1)\) such that \(q(r_1) \equiv 0\).
Compute the factorization\footnote{Polynomial factorizations are performed w.r.t. polynomials with coefficients in the fraction field of \(\mathbb{Q}⟨Σ^{⊕}⟩\) which is a computable field.} of \(g\) in the fraction field of \(\mathbb{Q}⟨Σ^{⊕}⟩\) and let \(\{q_1(X_1), \ldots, q_m(X_1)\}\) with \(m \geq 1\) be the set of all irreducible polynomials obtained thus as factors.
Because \(g(r_1) \equiv 0\), there exists an index \(j_0\) with \(1 \leq j_0 \leq m\) such that \(q_{j_0}(r_1) \equiv 0\) and \(q_j(r_1) \not\equiv 0\) for \(j \neq j_0\) and \(1 \leq j \leq m\).
Now we show how to find \(j_0\).
Using the operator \(R_k\) introduced in the beginning of Section~\ref{sec:decidability}, we  have that \(R_k(q_{j_0}(R_k(r_1))) \equiv 0\) for all \(k \geq 0\), while for each \(j \neq j_0\) there is always an index \(k_j\) such that \(R_{k_j}(q_{j}(R_{k_j}(r_1))) \not \equiv 0\).
Then, eventually an index \(j_0\) is always found.
Let \(q_{j_0}(X_1) = \frac{n_k}{d_k}X_1^{k} + \frac{n_{k-1}}{d_{k-1}}X_1^{k-1} + \ldots + \frac{n_0}{d_0}\) with \(k \geq 0, n_i,d_i \in \mathbb{Q}⟨Σ^{⊕}⟩\) and \(d_i\neq 0\) for all \(i\).
Let \(lcm(d_0, \ldots,d_k)\) denote the least common multiple of \(d_0, \ldots, d_k\) and define \(q(X_1) = lcm(d_0, \ldots,d_k) \cdot q_{j_0}(X_1) \).
Now \(q(X_1) \in \mathbb{Q}⟨Σ^{⊕}⟩⟨X_1⟩\) and this completes the algorithm.
\end{proof}

\begin{remark}
It is worth noting that, even though \(q(X_1)\) is an irreducible polynomial over \(K\), the fraction field of \(\mathbb{Q}⟨Σ^{⊕}⟩\), it might not be irreducible over \(\mathbb{Q}⟨Σ^{⊕}⟩\) since it might have a factorization consisting of a polynomial \(\widetilde{q}(X_1)\in \mathbb{Q}⟨Σ^{⊕}⟩⟨X_1⟩ \) of the same degree and one or more constant polynomials over \(\mathbb{Q}⟨Σ^{⊕}⟩\), i.e., polynomials of degree zero, that are not units in \(\mathbb{Q}⟨Σ^{⊕}⟩\).
However, since constant factors are not relevant for the result, we say that a polynomial over \(\mathbb{Q}⟨Σ^{⊕}⟩\) is \emph{irreducible} iff either no factorization exists, or, if there is one, then it is of the aforementioned form.
\end{remark}

\begin{theorem}
\label{thm:decision-proc}
Let \((G,W)\) be a cycle-free WCFG with \(W\) defined over \(\mathbb{Q}\). Then, it is decidable whether or not \((G,W)\) verifies the Parikh property.
\end{theorem}
\begin{proof}
Let \(S\) be the \(\mathbb{Q}⟨Σ^{⊕}⟩\)-algebraic system corresponding to \(G\) and let \(r_1\) be the first component of its strong solution.
Construct the irreducible polynomial \(q(X_1)\) with coefficients in \(\mathbb{Q}⟨Σ^{⊕}⟩\) as in Theorem~\ref{thm:q-construction}.
By Theorem~\ref{thm:linear-recurr}, we only need to check whether or not the equation \(q(X_1) = 0\) can be written as a linear equation of the form:
\mbox{\((1 - s)X_1 - t = 0,\)} with \(s,t \in \mathbb{Q}⟨Σ^{⊕}⟩\) and \((s, ε) = 0\).
Observe that the procedure given in Theorem~\ref{thm:q-construction} is complete, i.e., if the polynomial \(q\) obtained is not linear in \(X_1\) then there cannot exist a polynomial \(q_{\ell}(X_1)\) with coefficients in \(\mathbb{Q}⟨Σ^{⊕}⟩\) and linear in \(X_1\) such that \(q_{\ell}(r_1)\equiv 0\).
If it were the case, then \(q_\ell\) would be necessarily a factor of \(q\), and this contradicts the fact that \(q\) is irreducible over \(\mathbb{Q}⟨Σ^{⊕}⟩\).
Then, if \(q\) is not linear in \(X_1\), we conclude that \((G,W)\) does not satisfy the Parikh property.
Otherwise, \(q(X_1)\) can be rewritten as \(q(X_1) = (1-s)X_1-t\) with \(s,t \in \mathbb{Q}⟨Σ^{⊕}⟩\) and \((s, ε) = 0\), and we conclude that \((G,W)\) satisfies the Parikh property.
\end{proof}

Consider a WCFG \((G,W)\) with \(r_1\) the first component of the solution of its corresponding algebraic system.
Observe that, if the decision procedure returns a positive answer for \((G,W)\) then the polynomial \(q(X_1)\) constructed as in Theorem~\ref{thm:q-construction} is of the form:
\[q(X_1) = (s_0 - s_1)X_1 - t = 0\enspace ,\]
with \(s_0 \in \mathbb{Q}, s_0\neq 0\) and \(s_1,t \in \mathbb{Q}⟨Σ^{⊕}⟩\) with \((s_1,ε) = (t, ε) = 0\).
It follows that the algebraic system consisting of the equation:
\begin{equation}
\label{eq:linear-system}
X_1 = \frac{1}{s_0}s_1X_1 + \frac{1}{s_0}t \enspace ,
\end{equation}
has also \(r_1\) as solution.
Then a regular WCFG Parikh-equivalent to \((G,W)\) is the one corresponding to the regular algebraic system \eqref{eq:linear-system}.

Now we complete Examples~\ref{ex:example3} and~\ref{ex:example2} by following the decision procedure given in Theorem~\ref{thm:decision-proc} and giving the construction of a Parikh-equivalent regular WCFG (if exists).
Additionally, we give a third example.
\begin{example}
Consider the WCFG \((G,W)\) given in Example~\ref{ex:example3}.
Recall that its corresponding algebraic system \(S\) is given by the equation \(X = a\,X^2 + a\).
Let \(r\) be its strong solution.
Now we construct the irreducible polynomial \(q(X) \in \mathbb{Q}⟨\{a\}^{⊕}⟩⟨X⟩\) following the procedure given in Theorem~\ref{thm:q-construction}.
Let \(F = \{a\,X^2 - X + a\}\).
The reduced Groebner basis \(G\) of \(F\) w.r.t. reverse lexicographic ordering is (trivially) \(G = \{X^2 - \frac{ 1}{a}X + 1\}\).
Then the polynomial \(g \in G\) such that \(g \in K⟨X⟩\) where \(K\) is the fraction field of \(\mathbb{Q}⟨\{a\}^{⊕}⟩\), and \(g(r_1) \equiv 0\) is:
\[g(X) = X^2 - \frac{ 1}{a}X + 1 \enspace .\]
Note that this polynomial cannot be reduced into factors in the fraction field of \(\mathbb{Q}⟨\{a\}^{⊕}⟩\).
Multiplying \(g\) by \(a\), we get \(q(X) = aX^2 - X + a\) \(\in\mathbb{Q}⟨\{a\}^{⊕}⟩⟨X⟩\) and we conclude that \(q(X)\) is the irreducible polynomial described by Theorem~\ref{thm:q-existence}.
As \(q(X)\) is not linear we conclude that \((G,W)\) does not satisfy the Parikh property.
\end{example}
\begin{example}
Now consider the WCFG given in Example~\ref{ex:counterexample2} and its corresponding algebraic system \(S\).
We construct the irreducible polynomial \(q(X_2) \in \mathbb{Q}⟨\{a,\overline{a}\}^{⊕}⟩⟨X_2⟩\) following the procedure given in Theorem~\ref{thm:q-construction}.
Given \(F\), the set of polynomials in the left-hand sides of the equations of \(S\) after moving all monomials from right to left, we construct the reduced Groebner basis \(G\) of \(F\) w.r.t. reverse lexicographic ordering.
For clarity, we just show the polynomial \(g \in G\) such that \(g \in K⟨X_2⟩\) where \(K\) is the fraction field of \(\mathbb{Q}⟨\{a,\overline{a}\}^{⊕}⟩\), and verifies \(g(r_1) \equiv 0\):
\[g(X_2) = X_2 - \frac{1}{1 - (a + \overline{a})} \enspace .\]
This polynomial is linear so it is irreducible over the fraction field of \(\mathbb{Q}⟨\{a,\overline{a}\}^{⊕}⟩\).
Now we multiply \(g\) by \((1 - (a + \overline{a}))\) and thus obtain \(q(X_2) = (1 - (a + \overline{a}))X_2 - 1 \in \mathbb{Q}⟨\{a,\overline{a}\}^{⊕}⟩⟨X_2⟩\) which is the irreducible polynomial described by Theorem~\ref{thm:q-existence}.
Now we apply the decision procedure described in Theorem~\ref{thm:decision-proc}.
We observe that \(q\) can be written as follows:
\[q(X_2) = (1-s)X_2 - t = (1 - (a + \overline{a}))X_2 - 1 \enspace ,\]
with \((s, ε) =  0\).
Thus, we conclude that \((G,W)\) satisfies the Parikh property.
Finally, we give a regular Parikh-equivalent WCFG \((G_{\ell},W_{\ell})\).
The regular algebraic system:
\begin{equation}
\label{eq:example3}
(1 - (a + \overline{a}))X_2 - 1 = 0 \iff X_2 = (a + \overline{a}) X_2 + 1
\end{equation}
has \(r_1\) as solution.
Then, the WCFG \((G_{\ell},W_{\ell})\) corresponding to \eqref{eq:example3} is given by \(G_{\ell} = (\{X_2\}, \{a,\overline{a}\}, R_{\ell}, X_2)\) with \(R_{\ell}\) defined as:
\begin{align*}
\pi_1 = X_2 &\rightarrow aX_2\\
\pi_2 = X_2 &\rightarrow \overline{a}X_2 \\
\pi_3 = X_2 &\rightarrow \varepsilon
\end{align*}
and \(W_{\ell}\) defined over \((\mathbb{Q}, +, \cdot, 0, 1)\) as \(W_{\ell}(\pi_i) = 1\), for all \(i\).
Notice that \((G_{\ell}, W_{\ell})\) coincides with \((G_1, W_1)\) in Example~\ref{ex:counterexample2}.
\end{example}

\begin{example}
\label{ex:example1}
Consider the cycle-free WCFG \((G,W)\) given by \(G = (\{X_1,X_2\}, \{a,b\}, R, X_1)\) with \(R\) defined as follows:
\begin{align*}
X_1 &\rightarrow aX_2X_2\\
X_2 &\rightarrow bX_2 \mid a \enspace ,
\end{align*}
and the weight function \(W\) over \((\mathbb{Q}, +, \cdot, 0, 1)\) that assigns \(1\) to each production in the grammar.
The algebraic system \(S\) corresponding to \((G,W)\) is defined as follows:
\[
\begin{cases}
X_1 = a\,X_2^2 \\
X_2 = b\,X_2 + a \enspace .
\end{cases}
\]
Let \(\sigma = (r_1,r_2)\) be its strong solution.
Now we construct the irreducible polynomial \(q(X_1) \in \mathbb{Q}⟨\{a,b\}^{⊕}⟩⟨X_1⟩\) following the procedure given in Theorem~\ref{thm:q-construction}.
Let \(F = \{X_1 - a\,X_2^2, X_2 - b\,X_2 - a\}\).
The reduced Groebner basis\footnote{The Groebner basis \(G\) was computed using the \verb+groebner_basis+ method of the open-source mathematics software system \href{http://www.sagemath.org/}{SageMath}.} \(G\) of \(F\) w.r.t. lexicographic ordering is:
\[G = \left\{X_1 - \frac{ a^{3}}{b^{2} - 2 b + 1}, X_2 + \frac{a}{b - 1} \right\} \enspace .\]
Clearly, the polynomial \(g \in G\) such that \(g \in K⟨X_1⟩\) where \(K\) is the fraction field of \(\mathbb{Q}⟨Σ^{⊕}⟩\), and \(g(r_1) \equiv 0\) is:
\[g(X_1) = X_1 - \frac{ a^{3}}{b^{2} - 2 b + 1} \enspace .\]
This polynomial cannot be reduced into factors in the fraction field of \(\mathbb{Q}⟨Σ^{⊕}⟩\).
Now we multiply \(g\) by \((b^{2} - 2 b + 1)\) and thus obtain \(q(X_1) = (b^{2} - 2 b + 1)X_1 - a^3\) in \(\mathbb{Q}⟨Σ^{⊕}⟩⟨X_1⟩\) which is the irreducible polynomial described by Theorem~\ref{thm:q-existence}.
Now we apply the decision procedure described in Theorem~\ref{thm:decision-proc}.
We observe that \(q\) is linear in \(X_1\) and can be written as:
\[q(X_1) = (1-s)X_1 -t = (1 - ( 2 b - b^{2} )) X_1 - a^3 \enspace ,\]
with \((s, ε) =  0\).
Then we conclude that \((G,W)\) satisfies the Parikh property.
Note that this is the result expected as \((G,W)\) is nonexpansive.
Finally, we give a regular Parikh-equivalent WCFG \((G_{\ell},W_{\ell})\).
We know that the algebraic system:
\begin{equation}
\label{eq:example1}
X_1 = (2 b - b^{2}) X_1 + a^3
\end{equation}
has \(r_1\) as solution.
Then the WCFG \((G_{\ell},W_{\ell})\) corresponding to the regular system \eqref{eq:example1} is given by \(G_{\ell} = (\{X_1\}, \{a,b\}, R_{\ell}, X_1)\) with \(R_{\ell}\) defined as:
\begin{align*}
\pi_1 = X_1 &\rightarrow bX_1\\
\pi_2 = X_1 &\rightarrow  b^2X_1 \\
\pi_3 = X_1 &\rightarrow   a^3
\end{align*}
and \(W_{\ell}\) defined over \((\mathbb{Q}, +, \cdot, 0, 1)\) as:
\[W_{\ell}(\pi) = 
\begin{cases}
2 & \text{if }\pi = \pi_1\\
-1 & \text{if }\pi = \pi_2\\
1 & \text{if }\pi = \pi_3\\ 
\end{cases}
\enspace .\]
\end{example}

\section{Related Work}
\label{sec:relatedwork}
The problem of extending Parikh's Theorem to the weighted case has been significantly considered in the literature~\cite{Vijay17,Kuich1987,Luttenberger2016,Petre1999}.
Petre~\cite{Petre1999} establishes that the family of power series in commuting variables that can be generated by regular WCFGs is \emph{strictly} contained in that of the series generated by arbitrary WCFGs.
In this way, he shows that Parikh's Theorem does not hold in the weighted case.
It is well-known that the Parikh property holds in a commutative and  idempotent semiring~\cite{Vijay17,Kuich1987,Luttenberger2016}.
Luttenberger et al.~\cite{Luttenberger2016} deal with WCFGs where the weight of a word corresponds to its ambiguity (or commutative ambiguity when considering monomials instead of words) and they show that if a CFG is nonexpansive then its commutative ambiguity can be expressed by a weighted rational expression relying on the fact that all the parse trees of a nonexpansive CFG are of bounded dimension.
We used this fact to give a Parikh-equivalent regular WCFG construction, for a given nonexpansive WCFG defined over \emph{any} commutative semiring.
Baron and Kuich~\cite{Kuich1981} gave a similar characterization of nonexpansive grammars using rational power series to that of Luttenberger et al.
They also conjectured that an unambiguous WCFG is nonexpansive iff it has the Parikh property.
This conjecture appears to be false as evidenced by Example~\ref{ex:counterexample2}.
Bhattiprolu et al.~\cite{Vijay17} also show  that the class of \emph{polynomially ambiguous} WCFGs over the unary alphabet satisfies the property.
In the unary case, this class is strictly contained in the class of nonexpansive grammars (a proof is given in Appendix~\ref{sec:appendix-poly}).
Finally, our decision procedure relies on a result by Kuich and Saloma~\cite{KuichSalomaa1986} that decides if an \emph{algebraic} series in commuting variables with coefficients in \(\mathbb{Q}\) is rational.
To the best of our knowledge, the connection of this result to a decidability result for the Parikh property was only implicit.

\section{Conclusions and Further Work}
\label{sec:conclusions}
Note that from the theoretical point of view, our decision procedure can be applied to WCFGs over any arbitrary \emph{field}.
For arbitrary semirings, the decidability of the Parikh property remains open.
It would be interesting to tackle the question first in the unary case.
Finally, Theorem~\ref{thm:nexp-parikh} shows an equivalent characterization of the Parikh property.
Namely, the Parikh property holds for a WCFG \((G,W)\) iff there exists a Parikh-equivalent nonexpansive WCFG, i.e., iff \((G,W)\) is not \emph{inherently} expansive.
It is known that inherent expansiveness is undecidable in the noncommutative and unweighted case~\cite{Gruska1971}, but the question remains unsolved in the commutative case when weights are considered.

\appendix
\section{Proof of Theorem~\ref{thm:nexp-parikh}}
\label{sec:appendix-sec1}

First, we give the definitions we will use in this section.
Given a CFG \(G = (V,\Sigma, S, R{})\), define the \emph{degree} of \(G\) as \(max \{\len{\gamma⇃_{V}} : (X \rightarrow \gamma) ∈ R\} - 1\), where  \(\gamma⇃_{V}\) denotes the projection of \(\gamma\) onto the variables \(V\).
Given a production \(\pi = (X \longrightarrow \gamma) ∈ R\) and a position \(1 ≤ i ≤ \len{α}\), we define a \emph{derivation step} \(α \xRightarrow{\pi / i} \beta\) with \(α, \beta ∈ (Σ \cup V)^*\) iff \((α)_i = X\) and \(\beta = (α)_1… (α)_{i-1}\,\gamma\,(α)_{i+1}…(α)_{\len{α}}\). 
We omit the position \(i\) when it is not important.
We say that \(\alpha\) and \(\beta\) in \((Σ \cup V)^*\) are \emph{derivation sentences} of \(G\).
We define a \emph{derivation sequence} \(α_0 \xRightarrow{\pi_1} α_1 \xRightarrow{\pi_2} … \xRightarrow{\pi_n} α_n\) iff for every \(i ∈ \{1,…, n\}\), \(α_{i-1} \xRightarrow{\pi_{i}} α_{i}\) is a derivation step.
We call the derivation step \(α_{i-1} \xRightarrow{\pi_{i}} α_{i}\) the \emph{i-step} of the derivation sequence.
A derivation sequence \(ψ = α_0 \Rightarrow ⋯ \Rightarrow α_{n}\) of \(G\) has \emph{index} \(j\), denoted by \(idx(ψ)\), if for every \(i ∈ \{0, …, n\}\), no word \((α_i)⇃_{V}\) is longer than \(j\).
Now we define the dimension of a labeled tree as follows.

\begin{definition}
\label{def:dim}
Given a labeled tree $τ=c(τ_1,…,τ_n)$ (\(n \geq 0\)), the \emph{dimension} of $τ$ represented as \(\mathit{dim}(τ)\) is defined as follows:
\[
	\mathit{dim}(c(τ_1,…,τ_n))≝ \begin{cases}
		0 & \text{if } n=0 \\ 
		\mathit{dim}(τ_i)&\text{if } n>0 \land \vert\{ i \mid \forall j\colon \mathit{dim}(τ_j)≤\mathit{dim}(τ_i) \}\vert = 1 \\
		\mathit{dim}(τ_i) +1 &\text{if } n>0 \land \vert\{ i \mid \forall j\colon \mathit{dim}(τ_j)≤\mathit{dim}(τ_i) \}\vert > 1 
	\end{cases}
\]
\end{definition}

Now we present the proof of Theorem~\ref{thm:nexp-parikh}.
All the definitions, lemmas and theorems referred there can be found below the proof.

\subparagraph{Theorem~\ref{thm:nexp-parikh}}
Let \((G,W)\) be an arbitrary WCFG.  If \(G\) is nonexpansive then \((G,W)\) satisfies the Parikh property.
\begin{proof}
The proof is constructive.
For every nonexpansive WCFG \((G,W)\), we give a 2-step construction that results in a Parikh-equivalent regular WCFG \((G_{\ell}, W_{\ell})\).
The steps are:
\begin{enumerate}
\item construct a new WCFG \(\big(\dgrammar, \dweights\big)\), with \(k\in \mathbb{N}\), language-equivalent to \((G,W)\); and
\item construct a regular WCFG \((G_{\ell}, W_{\ell})\) Parikh-equivalent to \(\big(\dgrammar, \dweights\big)\). 
\end{enumerate}

The first part of the construction consists of building a new WCFG \(\big(\dgrammar, \dweights \big)\) (Definition~\ref{def:marked-WCFG} below), so-called \emph{at-most-k-dimension} WCFG of \((G,W)\), which is language-equivalent to the original and where grammar variables are annotated with information about the dimension of the parse trees that can be obtained from these variables.
Let us give an intuition on its construction.

For a given CFG \(G\) and \(k \in \mathbb{N}\) (the choice of \(k \in \mathbb{N}\) will be described later on), we define \(\dgrammar\) using the same construction as Luttenberger et al.~\cite{Luttenberger2016}.
They show how to construct, for a given CFG \(G\), a new grammar \(G^{⌈k⌉}\) with the property that \(\mathcal{T}_{G^{⌈k⌉}}\) corresponds to the subset of \(\mathcal{T}_G\) of trees of dimension at most \(k\).
They annotate each grammar variable with the superscript \([d]\) (resp. \(⌈d⌉\)) to denote that only parse trees of dimension exactly \(d\) (resp. at most \(d\)), where \(d \leq k\), can be obtained from these variables.
When constructing the grammar, they also consider those rules containing two or more variables in its right-hand side and distinguish which cases yield an increase of dimension.
We recall the construction of \(\dgrammar\) in Definition~\ref{def:marked-WCFG}.

To define the weight function \(\dweights\), we assign to each rule in \(\dgrammar\) the same weight as its corresponding version in \(G\) (note that for those rules in \(\dgrammar\) with no corresponding version in \(G\), i.e. the so-called \(e\)-rules, we assign the identity \(1_A\) with respect to \(\cdot\), where \(A\) denotes the weight domain).
Let us discuss the choice of \(k\) in \(\big(\dgrammar, \dweights\big)\).
Luttenberger et al.~\cite{Luttenberger2016} also show that if \(G\) is a nonexpansive CFG then the dimension of every parse tree in \(\mathcal{T}_G\) is bounded (Theorem~\ref{thm:michael}).
Moreover, the bound is at most the number of grammar variables of \(G\).
Then, for a given nonexpansive WCFG \((G,W)\), define \(k\) as this bound.
Because \(k\) is at most equal to the number of variables of \(G\), such a value is always found and consequently, the first part of the construction always terminates.
Finally, we show that the WCFG \(\big(\dgrammar, \dweights\big)\) is language-equivalent to \((G,W)\) (Lemma~\ref{lemma:parikh-equiv-marked}).

In the second part of the construction, we build a regular WCFG \((G_{\ell}, W_{\ell})\) that is Parikh-equivalent to \(\big( \dgrammar, \dweights\big)\).
Esparza et al.~\cite{Esparza2011} show that if the dimension of a parse tree is bounded by \(k\) then there exists a derivation sequence for the yield of the tree whose index is bounded by some affine function of \(k\) (Lemma~\ref{lemma:ipl}).
We rely on this result to define a special derivation policy over at-most-k-dimension WCFGs, for which we know the dimension of every parse tree is bounded by \(k\).
They are called \emph{lowest-dimension-first (LDF) derivations}.
We prove that, for every WCFG \(\big( \dgrammar,\dweights\big)\), the index of an LDF derivation sequence is always bounded by an affine function of \(k\) (Lemma~\ref{lemma:ldf-der}).
Then, each grammar variable of \((G_{\ell}, W_{\ell})\) represents each possible sentence (without the terminals) along an LDF derivation sequence of \(\big( \dgrammar,\dweights\big)\), and each grammar rule is intended to simulate an LDF derivation step of \(\big( \dgrammar,\dweights\big)\).
Because the number of variables in these sentences is bounded, the sets of variables and rules of \((G_{\ell}, W_{\ell})\) are necessarily finite.
A formal definition of the weighted regular \((G_{\ell}, W_{\ell})\) is given in Definition~\ref{def:regular-WCFG} .
Finally we show that \((G_{\ell}, W_{\ell})\) is Parikh-equivalent to \(\big( \dgrammar,\dweights\big)\) (Lemma~\ref{lemma:parikh-equiv-linear}) and this concludes the proof.
\end{proof}

Now we give the construction of the at-most-k-dimension WCFG \((\dgrammar, \dweights)\) for a given WCFG \((G,W)\) and \(k \in \mathbb{N}\).
For the construction of \(\dgrammar\), we rely on the one given by Luttenberger et al.~\cite{Luttenberger2016}.
\begin{definition}[The at-most-k-dimension WCFG]
\label{def:marked-WCFG}
Let \((G,W)\) be a WCFG with \(G = (V,Σ,S,R)\) and \(W\) defined over the commutative semiring \(A\), and let \(k \in \mathbb{N}\).
Define the \emph{at-most-k-dimension WCFG} \(\big(\dgrammar, \dweights\big)\) with \(\dgrammar = (\dvars,Σ,S^{⌈k⌉},\drules )\) of \((G,W)\) (with \(u_0, \ldots, u_n \in Σ^* \)) as follows:
\begin{itemize}
	\item The set \(\dvars\) of variables is given by 
	\[\{ X^{[d]}, X^{⌈d⌉} \mid  X ∈ V, 0 ≤ d ≤ k \}\enspace .\]
	\item The set \(\drules\) of production rules is given by
		\begin{enumerate}
			\item Linear rules:
			\begin{itemize}
				\item[\(\bullet\)] \(r_0(\pi) = \{X^{[0]} \rightarrow u_0\}\) for each rule  \(\pi = (X \rightarrow u_0) ∈ R\).
				\item[\(\bullet\)] \(r_1(\pi) = \{X^{[d]} \rightarrow u_0\,X_1^{[d]}\,u_1 \mid 0 ≤ d ≤ k\}\) for each rule \(\pi = (X \rightarrow u_0\,X_1\,u_1) ∈ R\).
			\end{itemize}
			\item Non-linear rules:\\
			For each rule \(\pi = (X \rightarrow u_0\,X_1\,u_1…u_{n-1}\,X_n\,u_n) ∈ R\)
			\begin{itemize}
				\item[\(\bullet\)] \(r_2(\pi) = \{X^{[d]} \rightarrow u_0\,Z_1\,u_1…u_{n-1}\,Z_n\,u_n \mid  \,1 ≤ d ≤ k,J \subseteq \{1, …, n\}\text{ with \mbox{\(|J| = 1\)}}: Z_i = X_{i}^{\,[d]} \text{ if }i \in J \text{, and }  Z_i = X_i^{⌈d-1⌉} \text{ for all }i \in \{1, \ldots, n\}\setminus J \}\)  and
				\item[\(\bullet\)] \(r_3(\pi) = \{X^{[d]} \rightarrow u_0\,Z_1\,u_1…u_{n-1}\,Z_n\,u_n \mid  1 ≤ d ≤ k, J \subseteq \{1, …, n\}\text{ with \mbox{\(|J| ≥ 2\)}}: Z_i = X_i^{[d-1]} \text{ for all } i ∈ J \text{ and } Z_i = X_i^{⌈d-1⌉} \text{ for all } i ∈ \{1,…,n\}\setminus J\}\).
			\end{itemize}
			\item \(e\)-rules:
			\begin{itemize}
				\item[\(\bullet\)] \(r_4 = \{X^{⌈d⌉} \rightarrow X^{[e]} \mid 0 ≤ e ≤ d ≤ k\}\).
			\end{itemize}
		\end{enumerate}
	\item The weight function \(\dweights\) is given by
	\begin{equation*}
	\dweights(φ) =
	    \begin{cases}
	      W(\pi) & \text{if } φ ∈ r_0(\pi) \text{ for some } \pi = (X \rightarrow u_0) ∈ R \\
	      W(\pi) & \text{if } φ ∈ r_1(\pi) \text{ for some } \pi = (X \rightarrow u_0\,X_1\,u_1) ∈ R \\
	      W(\pi) & \text{if } φ ∈ r_2(\pi) \cup r_3(\pi) \text{ for some } \pi = (X \rightarrow u_0\,Z_1\,u_1,…,u_{n-1}\,Z_n\,u_n) ∈ R \\
	      1_A        & \text{if } φ ∈ r_4
	    \end{cases}
	\end{equation*}
\end{itemize}
\end{definition}
We say that a variable \(Z \in \dvars\) is \emph{of dimension \(d\)} iff either \(Z = X^{⌈d⌉}\), or \(Z = X^{[d]}\), with \(X \in V\), and we denote it by \(dim(Z) = d\).
Define \(V^{(d)} ≝ \{Z \in \dvars \mid dim(Z) = d\}\), for each \(0 \leq d \leq k\).

\begin{theorem}[from Theorem 3.3 in~\cite{Luttenberger2016}]
\label{thm:michael}
Let \(G\) be a nonexpansive CFG with \(n\) variables.
Then there exists \(k\in \mathbb{N}\) with \(k \leq n\) such that every parse tree in \(\mathcal{T}_G\) has dimension at most \(k\).
\end{theorem}

\begin{lemma}
\label{lemma:parikh-equiv-marked}
\(⟦G⟧_W = ⟦\dgrammar⟧_{\dweights}\).
\end{lemma}
\begin{proof}
First recall that \(k\) corresponds to the nonnegative value such that every parse tree in \(G\) has dimension at most \(k\).
We want to show that there is a bijection \(\mu\) from \(𝒯_{\dgrammar}\) to \(𝒯_G\) that preserves the yield and the weight of each parse tree.

First, define \(𝒯^{\leq k}_G ≝ \{τ \mid τ \in 𝒯_G, dim(τ) \leq k \}\).
Luttenberger et al.~\cite{Luttenberger2016} prove that there is a bijection \(\mu\) from \(𝒯_{G^{⌈k⌉}}\) to \(𝒯^{\leq k}_G\) that preserves the yield of parse trees.
Roughly speaking, \(\mu\) contracts the edges corresponding to the \(e\)-rules and removes the superscripts from the labels of the trees.
Note that \(X^{⌈d⌉}\) can only be rewritten to \(X^{[e]}\) for some \(e \leq d\).
Then, contracting the corresponding edges cannot change the yield of the corresponding tree.
Furthermore, the rules of \(\dgrammar\) that rewrite the variable \(X^{[d]}\) are obtained from the rules of \(G\) that rewrite \(X\) by only adding a superscript.
Hence, by removing these annotations again, every tree \(τ \in 𝒯_{\dgrammar}\) is mapped by \(\mu\) to a tree in \(𝒯^{\leq k}_G\) with the same yield.
The complete proof is in Lemma 3.2 in~\cite{Luttenberger2016}.
Furthermore, because \(G\) is nonexpansive  we have that \(𝒯^{\leq k}_G = 𝒯_G\).
Thus, if \(G\) is nonexpansive, then \(\mu\) is a bijection from \(𝒯_{\dgrammar}\) to \(𝒯_G\) that preserves the yield of parse trees.

Now we show that \(\mu\) also preserves the weights of parse trees, i.e., for each \(τ \in 𝒯_{\dgrammar} : \dweights(τ) = W(\mu(τ))\).
We proceed by induction on the number of nodes of \(τ\).
In the base case, \(τ\) has one node, i.e., it has no children.
Then \(τ = φ \) with \(φ = X^{[0]} \rightarrow u_0\) and \(u_0 ∈ Σ^*\), and \(μ(τ) = φ'\) with \( φ' = X \rightarrow u_0\).
Then we have:
\begin{align*}
	\dweights(τ) &= \dweights(φ) &τ = φ\\
	&= W(φ')  &\text{by definition of } \dweights\\
	&= W(μ(τ)) &μ(τ) = φ'
\end{align*}
For the induction step, assume \(τ = φ(τ_1, …, τ_n)\) with \(n≥ 1\) and \(φ\) is a rule from the set \(r_i\) with \(i ∈ \{1,…, 5\}\) (see Definition~\ref{def:marked-WCFG}).
We distinguish three cases:
\begin{itemize}
	\item[\(\bullet\)] Assume \(φ ∈ r_1\)
	Then \(τ = φ(τ_1)\) and \(μ(τ) = φ'(\,μ(τ_1))\), with \(φ' = X \rightarrow u_0\,X_1\,u_1\) and \(u_0, u_1 ∈ Σ^*\).
	\begin{align*}
		\dweights(τ) &=  \dweights(φ) \cdot \dweights(τ_1) &\text{by definition of weight of } τ\\
		&= W(φ') \cdot W(μ(τ_1))  &\text{ by definition of } \dweights \text{ and induction hyp.}\\
		&= W(μ(τ)) &μ(τ) = φ'(μ(τ_1))
	\end{align*}
	\item[\(\bullet\)] Assume \(φ ∈ r_4\)
	Then \(τ = φ(τ_1)\) and \(μ(τ) = μ(τ_1)\).
	\begin{align*}
		\dweights(τ) &= \dweights(φ) \cdot \dweights(τ_1)  &τ = φ(τ_1)\\
		&= 1_A \cdot W(μ(τ_1))   &\text{by induction hyp
		and definition of } \dweights\\
		&= W(μ(τ)) & μ(τ) = μ(τ_1)\\
	\end{align*}
	\item[\(\bullet\)] Assume \(φ ∈ r_2 \cup r_3\)
	Then \(τ = φ(τ_1, …, τ_n)\) and \(μ(τ) = φ'(μ(τ_1),…,μ(τ_n))\) with \(φ' = X \rightarrow u_0\,X_1\,u_1…u_{n-1}\,X_n\,u_n\) and \(u_0, \ldots, u_n \in Σ^*\).
	\begin{align*}
		\dweights(τ) &= \dweights(φ) \prod^n_{i=1} \dweights(τ_i) &τ = φ(τ_1, …, τ_n)\\
		&= W(φ') \prod^{n}_{i=1}W(μ(τ_i))  &\text{by definition of } \dweights \text{ and induction hyp. }\\
		&= W(μ(τ)) &μ(τ) = φ'(μ(τ_1), …, μ(τ_n))
	\end{align*}

\end{itemize}
Finally, for each \(w \in Σ^{*}\):
\[⟦G⟧_W (w) =  W (w) =   \sum\limits_{\substack{ w = 𝒴(τ) \\τ \in 𝒯_{G} }} W(τ) =    \sum\limits_{\substack{ w = 𝒴(τ') \\ τ' \in 𝒯_{\dgrammar}\ }} \dweights(τ') =    \dweights(w) =   ⟦\dgrammar⟧_{\dweights}(w) \enspace .\]
\end{proof}

\begin{lemma}[from Lemma 2.2 in~\cite{Esparza2011}]
\label{lemma:ipl}
Let \(G\) be a CFG of degree \(m\) and let \(\tau \in \mathcal{T}_{G}\) with \(dim(\tau) \leq k\) and \(k \in \mathbb{N}\).
Then there is a derivation sequence for \(𝒴(\tau)\) of index at most \(km+1\).
\end{lemma}

Now we define a derivation policy over at-most-k-dimension WCFGs.
We will prove that this derivation policy satisfies Lemma~\ref{lemma:ipl} and thus the index of every derivation is bounded.
We call these derivations \emph{lowest-dimension-first (LDF) derivations}.

Intuitively, given a parse tree \(τ\) of an at-most-k-dimension WCFG, we define the LDF derivation sequence of \(τ\) by performing a depth-first traversal of \(τ\) where nodes in the same level of the tree are visited from lower to greater dimension and, if more than one node has the same dimension, then from left to right.
Recall that the dimension of a node corresponds to the dimension of the parse tree that it roots.

Before giving a formal definition, we introduce the following notation.
Given a derivation sequence \(\psi = \alpha_0 \Rightarrow \ldots \Rightarrow \alpha_n \) and \(\beta_0,\beta_1 \) (possibly empty) sequences of symbols and/or variables, we will denote by \(\beta_0\,\psi\,\beta_1\) the derivation sequence \( \beta_0\,\alpha_0\,\beta_1 \Rightarrow \ldots \Rightarrow \beta_0\,\alpha_n\,\beta_1 \).

\begin{definition}
Let \(\dgrammar\) be an at-most-k-dimension CFG as in Definition~\ref{def:marked-WCFG}.
Let \(τ = \pi(τ_1, \ldots, τ_n)\) be a parse tree of \(\dgrammar\).
Define the \emph{lowest-dimension-first (LDF) derivation sequence \(\psi\) of \(τ\)} inductively as follows:
\begin{itemize}
	\item If \(n = 0\), then \(\pi\) is of the form \(\pi = X^{[0]} \rightarrow u_0 \), and \(τ = \pi\).
	Then, the LDF derivation sequence of \(τ\) is:
	\[\psi = X^{[0]} \Rightarrow^{\pi}_{ldf} u_0 \enspace .\]
	\item If \(n\geq 1\), we distinguish the following cases:
		\begin{enumerate}
			\item If \(\pi \in r_1\), i.e., \(\pi\) is of the form \(\pi = X^{[d]} \rightarrow u_0 X_1^{[d]}u_1\) with \(0 \leq d \leq k\), and \(τ = \pi (τ_1)\).
			Then, the LDF derivation sequence of \(τ\) is:
			\[\psi = X^{[d]} \Rightarrow^{\pi}_{ldf} u_0X_1^{[d]}u_1 \Rightarrow^{}_{ldf} u_0\psi_1u_1 \enspace ,\]
			where \(\psi_1\) is the LDF derivation sequence of \(τ_1\).

			\item If \(\pi \in r_4\), i.e., \(\pi\) is of the form \(\pi = X^{⌈d⌉} \rightarrow  X^{[e]}\) with \(0 \leq e \leq d \leq k\), and \(τ = \pi (τ_1)\).
			Then, the LDF derivation sequence of \(τ\) is:
			\[\psi = X^{⌈d⌉} \Rightarrow^{\pi}_{ldf} X^{[e]} \Rightarrow^{}_{ldf} \psi_1 \enspace,\]
			where \(\psi_1\) is the LDF derivation sequence of \(τ_1\).

			\item If \(\pi \in r_2\), w.l.o.g., we assume that \(\pi\) is of the form:
			\[\pi = X^{[d]} \rightarrow  u_0X_1^{[d]}u_1X_{2}^{⌈d-1⌉}u_2\ldots u_{n-2}X_{n-1}^{⌈d-1⌉}u_{n-1}X_{n}^{⌈d-1⌉}u_n \enspace ,\]
	        with \(1\leq d \leq k\), and \(τ = \pi (τ_1, \ldots, τ_n )\).
	        Define, for each \(i\in\{2, \ldots, n\}\), the derivation sequence \(\tilde{\psi_i}\) as follows:
			\begin{align*}
			\tilde{\psi_i} ≝~ &u_0\,X_1^{[d]}\,u_1\,𝒴(τ_2)\,u_2\,\ldots\,𝒴(τ_{i-1})\,u_{i-1}\,X_{i}^{⌈d-1⌉}\,u_i\,\ldots\, u_{n-1}\,X_{n}^{⌈d-1⌉}u_n\\
			\Rightarrow^*_{ldf}~ &u_0\,X_1^{[d]}\,u_1\,𝒴(τ_2)\,u_2\,\ldots\,𝒴(τ_{i-1})\,u_{i-1}\,\psi_{i}\,u_i\,X_{i+1}^{⌈d-1⌉}\,u_{i+1}\,\ldots\, u_{n-1}\,X_{n}^{⌈d-1⌉ }u_n \enspace ,
			\end{align*}
			where \(\psi_i\) is the LDF derivation sequence of \(τ_i\). And define:
			\begin{align*}
			\tilde{\psi_1} ≝~ &u_0\,X_1^{[d]}\,u_1\,𝒴(τ_2)\,u_2\,\ldots\, u_{n-1}\,𝒴(τ_n)\,u_n\\
			\Rightarrow^*_{ldf}~ &u_0\,\psi_1\,u_1\,𝒴(τ_2)\,u_2\,\ldots\,u_{n-1}\,𝒴(τ_n)\,u_n\enspace ,
			\end{align*}
			where \(\psi_1\) is the LDF derivation sequence of \(τ_1\).
			Then the LDF derivation \(\psi\) of \(τ\) is:
			\[\psi = X^{[d]} \Rightarrow^{\pi}_{ldf} \tilde{\psi_2} \Rightarrow_{ldf} \ldots \Rightarrow_{ldf} \tilde{\psi_n} \Rightarrow_{ldf} \tilde{\psi_1} \enspace .\]

			\item If \(\pi \in r_3\), w.l.o.g., we assume that \(\pi\) is of the form:
			\[\pi = X^{[d]} \rightarrow  u_0X_1^{⌈d-1⌉}u_1X_{2}^{⌈d-1⌉}u_2\ldots u_{n-2}X_{n-1}^{[d-1]}u_{n-1}X_{n}^{[d-1]}u_n \enspace ,\]
	 		with \(1\leq d \leq k\), and \(τ = \pi (τ_1, \ldots, τ_n )\).
	 		Define, for each \(i\in\{1, \ldots, n\}\), the derivation sequence \(\tilde{\psi_i}\) as follows:
			\begin{align*}
			\tilde{\psi_i} ≝~ &u_0\,𝒴(τ_1)\,u_1\,𝒴(τ_2)\,u_2\,\ldots\,𝒴(τ_{i-1})\,u_{i-1}\,X_{i}^{⌈d-1⌉}\,u_i\,\ldots\, u_{n-1}\,X_{n}^{[d-1]}\,u_n\\
			\Rightarrow^*_{ldf}~ &u_0\,𝒴(τ_1)\,u_1\,𝒴(τ_2)\,u_2\,\ldots\,𝒴(τ_{i-1})\,u_{i-1}\,\psi_{i}\,u_i\,X_{i+1}^{⌈d-1⌉}\,u_{i+1}\,\ldots\, u_{n-1}\,X_{n}^{[d-1] }\,u_n\enspace ,
			\end{align*}
			where \(\psi_i\) is the LDF derivation sequence of \(τ_i\).
			The the LDF derivation \(\psi\) of \(τ\) is:
			\[\psi = X^{[d]} \Rightarrow_{ldf} \tilde{\psi_1} \Rightarrow_{ldf} \ldots \Rightarrow_{ldf} \tilde{\psi_n} \enspace .\]
		\end{enumerate}
\end{itemize}
\end{definition}
Note that, given a parse tree \(τ\) of \(\dgrammar\), the LDF derivation sequence of \(τ\) is uniquely defined.

\begin{lemma}
\label{lemma:ldf-der}
Let \(\dgrammar\) be an at-most-k-dimension CFG of degree \(m\) and \(τ ∈ \mathcal{T}_{\dgrammar}\) such that  \(dim(τ) ≤ k\). Then, the LDF derivation sequence of \(τ\) verifies \(idx(ψ) ≤ km+1\).
\end{lemma}
\begin{proof}

Let \(\dgrammar = (\dvars,Σ,S^{⌈k⌉},\drules)\).
We prove the more general statement: let \(m\) be the degree of \(\dgrammar\) and let \(τ ∈ \mathcal{T}_{\dgrammar}\) such that \(dim(\tau)\leq d\).
Then, the LDF derivation sequence \(\psi \) of \(τ\) satisfies \(idx(\psi) \leq dm+1\).
The proof goes by induction on the number of nodes of \(τ\).
In the base case, \(τ\) has one node, i.e., it has no children. Then, \(d= 0\) and the LDF derivation of \(τ\) is \(ψ = X^{[0]} \Rightarrow_{ldf}^{\pi} u_0\) with \(\pi = (X^{[0]} \rightarrow u_0) \in \drules\).
Clearly, the index of \(\psi\) is 1.

For the induction step, assume that \(τ = \pi(τ_1, \ldots, τ_n)\) with \(n \geq 1\).
We split the proof into the following four cases:
\begin{itemize}
	\item If \(\pi \in r_1\), then \(\pi\) is of the form \(\pi = X^{[d]} \rightarrow u_0 X_1^{[d]}u_1\) with \(0 \leq d \leq k\), and \(τ = \pi (τ_1)\) with \(dim(τ) \leq d\).
	By induction hypothesis, the LDF derivation sequence \(\psi_1\) of \(τ_1\) verifies \(idx(\psi_1) \leq (dm + 1)\).
	Then, the LDF derivation of \(τ\) is:
	\[\psi = X^{[d]} \Rightarrow_{ldf} u_0X_1^{[d]}u_1 \Rightarrow^{*}_{ldf} u_0\psi_1u_1,\]
	and verifies \(idx(\psi) \leq dm + 1\).

	\item If \(\pi \in r_4\), then \(\pi\) is of the form \(\pi = X^{⌈d⌉} \rightarrow  X^{[e]}\) with \(0 \leq e \leq d \leq k\), and \(τ = \pi (τ_1)\) with \(dim(τ) \leq d\).
	By induction hypothesis, the LDF derivation sequence \(\psi_1\) of \(τ_1\) s.t. \(idx(\psi_1) \leq (em + 1)\).
	Then, the LDF derivation of \(τ\) is:
	\[\psi = X^{⌈d⌉} \Rightarrow_{ldf} X^{[e]} \Rightarrow^{*}_{ldf} \psi_1,\]
	and verifies \(idx(\psi) \leq em + 1 \leq dm +1\).

	\item If \(\pi \in r_2\), then, w.l.o.g., \(\pi\) is of the form:
	 \[\pi = X^{[d]} \rightarrow  u_0X_1^{[d]}u_1X_{2}^{⌈d-1⌉}u_2\ldots u_{n-2}X_{n-1}^{⌈d-1⌉}u_{n-1}X_{n}^{⌈d-1⌉}u_n,\]
	 with \(1\leq d \leq k\), and \(τ = \pi (τ_1, \ldots, τ_n )\) with \(dim(τ) \leq d\).
	By induction hypothesis, for each \(i\in\{2, \ldots, n\}\), there is a derivation \(\psi_i\) of \(τ_i\) s.t. \(idx(\psi_i) \leq ((d-1)m + 1)\), and there is a derivation \(\psi_1\) for \(τ_1\) s.t. \(idx(\psi_1) \leq dm + 1\).
	Now, define, for each \(i\in\{2, \ldots, n\}\), the derivation sequence \(\tilde{\psi_i}\) as follows:
	\begin{align*}
	\tilde{\psi_i} ≝~ &u_0\,X_1^{[d]}\,u_1\,𝒴(τ_2)\,u_2\,\ldots\,𝒴(τ_{i-1})\,u_{i-1}\,X_{i}^{⌈d-1⌉}\,u_i\,\ldots\, u_{n-1}\,X_{n}^{⌈d-1⌉}u_n\\
	\Rightarrow^*_{ldf}~ &u_0\,X_1^{[d]}\,u_1\,𝒴(τ_2)\,u_2\,\ldots\,𝒴(τ_{i-1})\,u_{i-1}\,\psi_{i}\,u_i\,X_{i+1}^{⌈d-1⌉}\,u_{i+1}\,\ldots\, u_{n-1}\,X_{n}^{⌈d-1⌉ }u_n \enspace .
	\end{align*}
	And define:
	\begin{align*}
	\tilde{\psi_1} ≝~ &u_0\,X_1^{[d]}\,u_1\,𝒴(τ_2)\,u_2\,\ldots\, u_{n-1}\,𝒴(τ_n)\,u_n\\
	\Rightarrow^*_{ldf}~ &u_0\,\psi_1\,u_1\,𝒴(τ_2)\,u_2\,\ldots\,u_{n-1}\,𝒴(τ_n)\,u_n\enspace .
	\end{align*}
	Then, the LDF derivation \(\psi\) of \(τ\) is:
	\[\psi = X^{[d]} \Rightarrow_{ldf}^{\pi} \tilde{\psi_2} \Rightarrow_{ldf} \ldots \Rightarrow_{ldf} \tilde{\psi_n} \Rightarrow_{ldf} \tilde{\psi_1} \enspace .\]

	Observe that \(n-1 \leq m\) where \(n\) is the number of variables occurring in the right-hand side of \(\pi\) and \(m\) is the degree of \(\dgrammar\). For each \(i \in \{2, \ldots, n\}\), the index of \(\tilde{\psi_i}\) is at most \((d-1)m + 1 + (n-i) \leq dm + 1\). On the other hand, the index of \(\tilde{\psi_1}\) is at most \(dm + 1\). Then, performing the derivation steps of \(\psi\) in the order shown above we have that \(idx(\psi) \leq dm + 1\).

	\item If \(\pi \in r_3\), then, w.l.o.g., \(\pi\) is of the form:
	 \[\pi = X^{[d]} \rightarrow  u_0X_1^{⌈d-1⌉}u_1X_{2}^{⌈d-1⌉}u_2\ldots u_{n-2}X_{n-1}^{[d-1]}u_{n-1}X_{n}^{[d-1]}u_n,\]
	 with \(1\leq d \leq k\), and \(τ = \pi (τ_1, \ldots, τ_n )\) with \(dim(τ) \leq d\).
	By induction hypothesis, for each \(i\in\{1, \ldots, n\}\), there is a derivation \(\psi_i\) of \(τ_i\) s.t. \(idx(\psi_i) \leq ((d-1)m + 1)\).
	Now, define, for each \(i\in\{1, \ldots, n\}\), the derivation sequence \(\tilde{\psi_i}\) as follows:
	\begin{align*}
	\tilde{\psi_i} ≝~ &u_0\,𝒴(τ_1)\,u_1\,𝒴(τ_2)\,u_2\,\ldots\,𝒴(τ_{i-1})\,u_{i-1}\,X_{i}^{⌈d-1⌉}\,u_i\,\ldots\, u_{n-1}\,X_{n}^{[d-1]}\,u_n\\
	\Rightarrow^*_{ldf}~ &u_0\,𝒴(τ_1)\,u_1\,𝒴(τ_2)\,u_2\,\ldots\,𝒴(τ_{i-1})\,u_{i-1}\,\psi_{i}\,u_i\,X_{i+1}^{⌈d-1⌉}\,u_{i+1}\,\ldots\, u_{n-1}\,X_{n}^{[d-1] }\,u_n\enspace .
	\end{align*}
	Then, the LDF derivation \(\psi\) of \(τ\) is:
	\[\psi = X^{[d]} \Rightarrow_{ldf}^{\pi} \tilde{\psi_1} \Rightarrow_{ldf} \ldots \Rightarrow_{ldf} \tilde{\psi_n} \enspace .\]
	For each \(i \in \{1, \ldots, n\}\), the index of \(\tilde{\psi_i}\) is at most \((d-1)m + 1 + (n-i) \leq dm + 1\). It follows that \(idx(\psi) \leq dm + 1\).

\end{itemize}
\end{proof}
Given a derivation sentence \(\alpha \in \big( \Sigma \cup \dvars \big)^*\) of an at-most-k-dimension CFG, define \(\mathcal{LDF}(\alpha) ≝ α⇃_{\Sigma}\,α⇃_{V^{(0)}}\,α⇃_{V^{(1)}}\ldots α⇃_{V^{(k)}} \) and \(\mathcal{LDF}_{\dvars}(\alpha) ≝ (\mathcal{LDF}(\alpha))⇃_{\dvars}\).
Now we define a regular \((G_{\ell}, W_{\ell})\) that is Parikh-equivalent to \(\big(\dgrammar, \dweights\big)\) in a similar way to Bhattiprolu et al.~\cite{Vijay17}.

\begin{definition}[Regular WCFG for \(\big(\dgrammar, \dweights\big)\)]
\label{def:regular-WCFG}
Let \(\big(\dgrammar, \dweights\big)\) be an at-most-k-dimension WCFG  with \(\dgrammar = (\dvars,Σ,S^{⌈k⌉},\drules)\) and degree \(m\), and \(\dweights\) defined over the commutative semiring \(A\).
Define the WCFG \((G_{\ell}, W_{\ell})\) with \(G_{\ell} = (V_{\ell},Σ,S_{\ell},R_{\ell} )\) as follows:
\begin{itemize}
	\item Each variable in \(V_{\ell}\) corresponds to a sequence \(\alpha \in \big ( \dvars \big )^{km + 1} \) where \(( \dvars \big )^{km + 1}\) denotes the set \(\{w \mid w \in ( \dvars \big )^*, |w| \leq km +1 \}\), and we denote it by \(⟨\alpha⟩\).
	Formally, 
	\[V_{\ell}≝ \{ ⟨\alpha⟩ \mid  \alpha \in \big ( \dvars \big )^{km + 1} \}\enspace .\]
	\item The initial variable is defined as \(S_{\ell} ≝ S^{⌈k⌉}\).
	\item For each rule \(\pi = (X \longrightarrow \gamma) ∈ \drules\) define \(\pi^{α}≝ (⟨X\,α⟩ \longrightarrow \gamma⇃_{\Sigma}\,⟨\mathcal{LDF}_{\dvars}(\gamma)\,α⟩)\).
	The set \(R_{\ell}\) of rules is given by
		\[\{\pi^{α} \mid \pi = (X \longrightarrow \gamma) ∈ \drules \text{ and }⟨X\alpha⟩,⟨\mathcal{LDF}_{\dvars}(\gamma)\,\alpha⟩ \in V_{\ell} \}\enspace .\]
	\item The weight function \(W_{\ell}\) is given by
	\[W_{\ell}(\pi^{\alpha}) ≝ \dweights(\pi) \text{ for all } \pi^α ∈ R_{\ell} \enspace .\]
\end{itemize}
\end{definition}
\begin{lemma}
\label{lemma:parikh-equiv-linear}
\(Pk⟦\dgrammar⟧_{\dweights} = Pk⟦G_{\ell}⟧_{W_{\ell}} \).
\end{lemma}
\begin{proof}
First, we give the definitions and notation we will use in this proof.
For convenience, we will give an alternative definition of the weight of a word using derivation sequences.
Recall that we assume that the derivation policy of a grammar defines for each parse tree one unique derivation sequence.
Given a CFG \(G\) and \(w \in \Sigma^*\), define by \(parse_G(w)\) as the subset  of all derivations of \(G\) that yield to \(w \in Σ^{*}\).
Then, define for each derivation sequence \(\psi = \alpha_0 \Rightarrow^{\pi_1} \alpha_1 \Rightarrow^{\pi_2} \ldots \Rightarrow^{\pi_n} \alpha_n\) of \(G\) the \emph{weight of \(\psi\)} as follows:
\[W(\psi) ≝ \prod^{n}_{i=1} W(\pi_i) \enspace .\]
Finally, define for each \(w \in Σ^{*}\), 
\[W(w) ≝ \sum_{\psi \in parse_G(w)} W(\psi) \enspace .\]
If \(parse_G(w) = \emptyset\) then \(W(w) ≝ 0_A\).
Given a derivation sequence \(\psi = \alpha_0 \Rightarrow \ldots \Rightarrow \alpha_n \) and \(\beta_0,\beta_1 \) (possibly empty) sequences of symbols and/or variables, recall that the notation \(\beta_0\,\psi\,\beta_1\) denotes the derivation sequence \( \beta_0\,\alpha_0\,\beta_1 \Rightarrow \ldots \Rightarrow \beta_0\,\alpha_n\,\beta_1 \).
When \(\beta_1 \in V_{\ell}\) and thus \(\beta_{1} =  ⟨\beta'_1⟩\) for some \(\beta'_1 \in \big ( \dvars \big )^{km + 1}\), and \(\alpha_i\) is of the form  \(w⟨\alpha'_i⟩\) with  \(\alpha'_i \in V_{\ell}\) and \(w \in \Sigma^*\), for some \(i\in \{1, \ldots, n\}\), then \(\beta_0\alpha_i\beta_1\) denotes \(\beta_0 w ⟨\alpha'_i \beta'_1⟩\).

We claim that there exists a one-to-one correspondence \(f\) that maps each LDF derivation sequence of \((\dgrammar, \dweights)\) into a derivation sequence of \((G_{\ell}, W_{\ell})\) that preserves the Parikh images and the weights between derivations.
Formally, there exists a one-to-one correspondence \(f\) such that for each LDF derivation sequence \(\psi = X^{[d]} \Rightarrow^*_{ldf} w\) with \(w \in \Sigma^*\) of \((\dgrammar, \dweights)\), \(f(\psi) = ⟨X^{[d]}⟩ \Rightarrow^* w'\) with \(w' \in \Sigma^*\) is a derivation sequence of \((G_{\ell}, W_{\ell})\) with the following properties:
\begin{enumerate}
	\item \(\lbag w\rbag = \lbag w' \rbag\) and,
	\item \(\dweights(\psi) = W_{\ell}(f(\psi))\).
\end{enumerate}

We now give an inductive definition of \(f\).
Along this definition we will prove inductively that:
\begin{inparaenum}[\upshape(\itshape i\upshape)]
\item \(f\) is an injective function from LDF derivation sequences of \((\dgrammar, \dweights)\) to derivation sequences in \((G_{\ell}, W_{\ell})\); and
\item properties 1. and 2. above hold.  
\end{inparaenum}

Let \(\psi\) be a LDF derivation of \((\dgrammar, \dweights)\).
\begin{enumerate}
	\item If \(\psi\) is a 1-step derivation sequence then \(\psi = X^{[0]}\Rightarrow^{\pi}_{ldf} u_0\) with \(\pi \in r_0\).
	Then, define \(f(\psi) ≝ ⟨X^{[0]}⟩ \Rightarrow^{\pi^{\varepsilon}} u_0\).

    Note that \(f(\psi)\) is a one-step derivation sequence that uses the rule \((⟨X^{[0]}⟩\rightarrow u_0) \in R_{\ell}\). 
    It follows that \(f\) defines uniquely a derivation sequence of \((G_{\ell}, W_{\ell})\) for \(\psi\).
	Note that property 1. holds trivially.
	By definition of \((G_{\ell}, W_{\ell})\), we have that:
	\[W_{\ell}(f(\psi)) = W_{\ell}(⟨X^{[0]}⟩\rightarrow u_0) = \dweights(X^{[0]}\rightarrow u_0) = \dweights(\psi) \enspace .\]

	\item If \(\psi\) is a \(n\)-step derivation sequence (with \(n>1\)), then we have the following cases:
	\begin{itemize}
		\item If \(\psi = X^{[d]} \Rightarrow^{\pi}_{ldf} u_0 X_{1}^{[d]} u_1 \Rightarrow^{*}_{ldf} u_0 \psi'u_1\) where \(\pi \in r_1\) and \(\psi' = X_{1}^{[d]} \Rightarrow^{*}_{ldf} w\) with \(w \in \Sigma^*\).
		Then, define \(f(\psi) ≝ ⟨X^{[d]}⟩ \Rightarrow^{\pi^{\varepsilon}} u_0u_1⟨X_1^{[d]}⟩ \Rightarrow^{*} u_0u_1f(\psi')\).

		Note that the first step in the derivation \(f(\psi)\) uses the rule \((⟨X^{[d]}⟩ \rightarrow u_0u_1⟨X_1^{[d]}⟩) \in R_{\ell}\).
    	Relying on this and the hypothesis of induction, \(f\) defines uniquely a derivation sequence of \((G_{\ell}, W_{\ell})\) for \(\psi\).
		By hypothesis of induction, it is easy to check that property 1. holds.
		Finally, using the hypothesis of induction and the definition of \((G_{\ell}, W_{\ell})\) we have:
		\[W_{\ell}(f(\psi)) = W_{\ell}(⟨X^{[d]}⟩ \rightarrow u_0u_1⟨X_1^{[d]}⟩)\cdot W_{\ell}(f(\psi'))  = \dweights(X^{[d]} \rightarrow u_0X_1^{[d]}u_1) \cdot \dweights(\psi') = \dweights(\psi) \enspace .\]

		\item If \(\psi = X^{⌈d⌉} \Rightarrow^{\pi}_{ldf} X^{[e]} \Rightarrow^{*}_{ldf} \psi'\) where \(\pi \in r_2\) and \(\psi' = X_{1}^{[e]} \Rightarrow^{*}_{ldf} w\) with \(w \in \Sigma^*\).
		Then, define \(f(\psi) ≝ ⟨X^{⌈d⌉}⟩ \Rightarrow^{\pi^{\varepsilon}} ⟨X^{[e]}⟩ \Rightarrow^{*} f(\psi')\).

		Note that the first step in the derivation \(f(\psi)\) uses the rule \((⟨X^{⌈d⌉}⟩ \rightarrow ⟨X^{[e]}⟩) \in R_{\ell}\).
    	Relying on this and the hypothesis of induction, \(f\) defines uniquely a derivation sequence of \((G_{\ell}, W_{\ell})\) for \(\psi\).
		By hypothesis of induction, property 1. holds trivially.
		Finally, using the hypothesis of induction and the definition of \((G_{\ell}, W_{\ell})\) we have:
		\[W_{\ell}(f(\psi)) = W_{\ell}(⟨X^{⌈d⌉}⟩ \rightarrow ⟨X^{[e]}⟩)\cdot W_{\ell}(f(\psi'))  = \dweights(X^{⌈d⌉} \rightarrow X^{[e]}) \cdot \dweights(\psi') = \dweights(\psi) \enspace .\]

		\item Finally, assume w.l.o.g, that \(\psi\) has the form: 
		\begin{align*}
		\psi = X^{[d]} &\Rightarrow^{\pi}_{ldf}  u_0Z_1u_1\ldots u_{n-2}Z_{n-1} u_{n-1}Z_nu_n \\ 
		&\Rightarrow^{*}_{ldf}  u_0Z_1u_1\ldots u_{n-2}Z_{n-1} u_{n-1}\psi'_n u_n \\
		&\Rightarrow^{*}_{ldf}  u_0Z_1u_1\ldots u_{n-2}\psi'_{n-1} u_{n-1} w_n u_n \Rightarrow^{*}_{ldf} \ldots \\
		&\Rightarrow^{*}_{ldf}  u_0 w_1 u_1 w_2 u_2 \ldots u_{n-1} w_n u_n \enspace ,
		\end{align*}
		where \(\pi \in r_2 \cup r_3\), and for each \(i \in\{1, \ldots, n\}\) \(\psi'_i = Z_i \Rightarrow^{*}_{ldf} w_i\)  with \(w_i \in \Sigma^*\) for all \(i\).
		Then, define
		\begin{align*}
		f(\psi) ≝ ⟨X^{[d]}⟩&\Rightarrow^{\pi^{\varepsilon}} u_0 u_1 u_2 \ldots u_n ⟨\mathcal{LDF}(Z_1\ldots Z_n)⟩\\
		&\Rightarrow^{*} u_0 u_1 u_2 \ldots u_n f(\psi'_n)⟨\mathcal{LDF}(Z_1\ldots Z_{n-1})⟩ \\
		&\Rightarrow^{*} u_0 u_1 u_2 \ldots u_n w'_n f(\psi'_{n-1})⟨\mathcal{LDF}(Z_1\ldots Z_{n-2})⟩\Rightarrow^{*} \ldots\\
		&\Rightarrow^{*}  u_0 u_1 u_2 \ldots u_n w'_n w'_{n-1}\ldots w'_1 \enspace ,
		\end{align*}
		where each \(w'_i \in \Sigma^*\) corresponds to the word generated by each \(f(\psi'_i)\) inductively.
		Note that the first step in the derivation \(f(\psi)\) uses a rule of the form \((X^{[d]} \rightarrow  u_0 u_1 u_2 \ldots u_n ⟨\mathcal{LDF}(Z_1\ldots Z_n)⟩)\) which according to the definition of \((G_{\ell}, W_{\ell})\) defines a rule in \(R_{\ell}\).
    	Relying on this and the hypothesis of induction, \(f\) defines uniquely a derivation sequence of \((G_{\ell}, W_{\ell})\) for \(\psi\).
		It is easy to see property 1. holds since, by hypothesis of induction, each \(w'_i\) satisfies \(\lbag w'_i \rbag = \lbag w_i \rbag\).
		Finally, using the hypothesis of induction and the definition of \((G_{\ell}, W_{\ell})\) we have:
		\begin{align*}
		W_{\ell}(f(\psi)) &= W_{\ell}(X^{[d]} \rightarrow  u_0 u_1 \ldots u_n ⟨\mathcal{LDF}(Z_1\ldots Z_n)⟩)\cdot  \prod^n_{i=1} W_{\ell}(f(\psi'_i))   \\
		&= \dweights(X^{[d]} \rightarrow  u_0Z_1u_1\ldots u_{n-1}Z_nu_n ) \cdot \prod^n_{i=1} \dweights(\psi'_i) = \dweights(\psi) \enspace .
		\end{align*}
	\end{itemize}
\end{enumerate}

Finally, the fact that \(f\) is a surjective function follows from its construction.
First, note that each rule in \(R_{\ell}\) is intended to simulate a LDF derivation step of \((\dgrammar, \dweights)\), while each variable in \(V_{\ell}\) represents a derivation sequence of \((\dgrammar, \dweights)\).
On the other hand, the reader can check  that, in each case, the definition of \(f\) intends to simulate a LDF derivation of \((\dgrammar, \dweights)\) using the convenient definition of rules and variables of \((G_{\ell}, W_{\ell})\).
It follows that every derivation sequence of \((G_{\ell}, W_{\ell})\) is covered by the image of \(f\).

Relying on the definition of \(f\) and its properties, the following equalities hold.
For each \(v \in Σ^{⊕}\):
\begin{align*}
      Pk ⟦ \dgrammar ⟧_{\dweights}(v) &= \sum_{v = \lbag w \rbag} ⟦ \dgrammar ⟧_{\dweights}(w)  &\text{ by definition of Parikh image} \\
      &=  \sum_{v = \lbag w \rbag} \sum_{ψ ∈ parse_{\dgrammar}(w)}\dweights(ψ) & \text{ by definition of semantics}\\
      &= \sum_{v = \lbag w \rbag} \sum_{ψ ∈ parse_{\dgrammar}(w)} W_{\ell}(f(ψ)) & f\text{ satisfies prop. 2.}\\
      &= \sum_{v = \lbag w' \rbag} \sum_{ψ' ∈ parse_{G_{\ell}}(w')} W_{\ell}(ψ') & f \text{ is a bijection and satisfies prop. 1.}\\
      &= \sum_{v =  \lbag w' \rbag} ⟦ G_{\ell} ⟧_{W_{\ell}}(w')  &\text{ by definition of semantics}\\
      &=Pk ⟦ G_{\ell} ⟧_{W_{\ell}}(v). &\text{ by definition of Parikh image}
\end{align*}
\end{proof}

\section{Counterexample in the unary case}
\label{sec:appendix-counterexample}

In Example~\ref{ex:counterexample2} we show an expansive WCFG \((G_2, W_2)\) over the alphabet \(\{a,\overline{a}\}\) that satisfies the Parikh property.
Now we show, by means of the following example, that there also exists an expansive WCFG \((G, W)\) over the alphabet \(\{a\}\) that satisfies the Parikh property. 

\begin{example}
The idea behind this example is to use the definition of \((G_2, W_2)\) from Example~\ref{ex:counterexample2} (a complete definition is given in Example~\ref{ex:example2}) and replace each occurrence of the alphabet symbol \(\overline{a}\) in the rules of \((G_2, W_2)\) by \(a\).
Thus, define the WCFG \((G,W)\) where \(G = (\{X,\overline{D},D,Y,Z\}, \{a\}, X, R)\), \(R\) is given by:
\begin{align*}
X &\rightarrow D \mid \overline{D} & \overline{D} &\rightarrow D\,a\,Y \mid D\,a\,Z & Z &\rightarrow D\,a\,Z \mid D \enspace ,
\\
D &\rightarrow a\,D\,a\,D \mid \varepsilon & Y &\rightarrow a\,Y \mid \varepsilon
\end{align*}
and the weight function \(W\) is defined over \((\mathbb{N}, +, \cdot, 0, 1)\) and assigns \(1\) to each production in the grammar except from the rule \(Y \rightarrow a\,Y\) which is assigned weight \(2\).
Notice that we preferred to assign weight \(2\) to the later rule instead of adding two copies each of weight \(1\).
Recall that \(Pk ⟦ G_2 ⟧_{{W}_2} = (a + \overline{a})^*\).
Now, relying on our construction of \((G,W)\), we have that \(Pk ⟦ G ⟧_{W}\) is the formal power series that results from replacing each \(\overline{a}\) by \(a\) in the series \(Pk ⟦ G_2 ⟧_{{W}_2}\).
Thus, we obtain that \(Pk ⟦ G ⟧_{W} = (a+a)^* = (2a)^*\).
The reader can check that the formal power series \((2a)^*\) corresponds to the Parikh image of the regular WCFG \((G_{\ell},W_{\ell})\) where \(G_{\ell}\) is defined as \(G_{\ell} = (\{X\}, \{a\}, X, \{X \rightarrow aX,\,X \rightarrow \varepsilon\})\) and the weight function \(W_{\ell}\) is defined over \((\mathbb{N}, +, \cdot, 0, 1)\) and assigns \(2\) to the rule \(X \rightarrow aX\) and \(1\) to the rule \(X \rightarrow \varepsilon\).
\end{example}

\begin{remark}
Let us check that \((G,W)\) has the Parikh property using the decision procedure presented in Section~\ref{sec:decidability}.
The algebraic system \(S\) corresponding to \((G,W)\) consists of the following equations:
\begin{align*}
X &= D + \overline{D} & \overline{D} &= D\,a\,Y + D\,a\,Z & Z &= D\,a\,Z + D \enspace .
\\
D &= a\,D\,a\,D + 1 & Y &= 2a\,Y + 1
\end{align*}
Let \(\sigma = (r_1, r_2, r_3, r_4, r_5)\) be its strong solution where \(r_1\) corresponds to the solution for the initial variable \(X\).

Now we construct the irreducible polynomial \(q(X) \in \mathbb{Q}⟨\{a\}^{⊕}⟩⟨X⟩\) following the procedure given in Theorem~\ref{thm:q-construction}.
Let \(F = \{X - D - \overline{D}, \overline{D} - DaY -DaZ, Z - DaZ - D, D - aDaD - 1, Y-2aY-1\}\).
The reduced Groebner basis\footnote{The Groebner basis \(G\) was computed using the \verb+groebner_basis+ method of the open-source mathematics software system \href{http://www.sagemath.org/}{SageMath}.} \(G\) of \(F\) w.r.t. lexicographic ordering is:
\begin{align*}
G = & \left\{X - \frac{1}{1-2a},\,D + \overline{D} - \frac{1}{1 - 2a},\,Z - \left(\frac{a}{1 - 2a}\right) \overline{D} - \frac{1 - 3a}{4a^{2} - 4a + 1},\,Y - \frac{1}{1 - 2a},\right. \\
& \left.\overline{D}^{2} + \left(\frac{2a^{2} + 2a - 1}{2a^{3} - a^{2}}\right) \overline{D} + \frac{5a - 2}{4a^{3} - 4a^{2} + a} \right\} \enspace .
\end{align*}
Clearly, the polynomial \(g \in G\) such that \(g \in K⟨X⟩\) where \(K\) is the fraction field of \(\mathbb{Q}⟨\{a\}^{⊕}⟩\), and \(g(r_1) \equiv 0\) is:
\[g(X) = X - \frac{1}{1-2a} \enspace .\]

This polynomial cannot be reduced into factors in the fraction field of \(\mathbb{Q}⟨\{a\}^{⊕}⟩\).
Now we multiply \(g\) by \((1 - 2a)\) and thus obtain \(q(X) = (1 - 2a)X - 1\) in \(\mathbb{Q}⟨\{a\}^{⊕}⟩⟨X⟩\).
We have that \(q(X) = (1 - 2a)X - 1\) is the irreducible polynomial described by Theorem~\ref{thm:q-existence}.
Now we apply the decision procedure given in Theorem~\ref{thm:decision-proc}.
We observe that \(q\) is linear in \(X\) and can be written as:
\[q(X) = (1-s)X -t = (1 - 2a) X - 1 \enspace ,\]
with \((s, ε) = 0\).
Then, we conclude that \((G,W)\) satisfies the Parikh property as expected.
Let us give the regular Parikh-equivalent WCFG \((G_{\ell},W_{\ell})\).
We know that the algebraic system:
\begin{equation}
\label{eq:example1}
X = 2a X + 1
\end{equation}
has \(r_1\) as solution.
Then the WCFG \((G_{\ell},W_{\ell})\) corresponding to the regular system \eqref{eq:example1} is given by \(G_{\ell} = (\{X\}, \{a\}, R_{\ell}, X)\) with \(R_{\ell}\) defined as:
\begin{align*}
\pi_1 = X &\rightarrow aX\\
\pi_2 = X &\rightarrow  \varepsilon \\
\end{align*}
and \(W_{\ell}\) defined over \((\mathbb{Q}, +, \cdot, 0, 1)\) as:
\[W_{\ell}(\pi) = 
\begin{cases}
2 & \text{if }\pi = \pi_1\\
1 & \text{if }\pi = \pi_2\\
\end{cases}
\enspace .\]
Note that \((G_{\ell}, W_{\ell})\) coincides with the regular WCFG given in the example.
\end{remark}

\section{A decision procedure for the Parikh property over the rationals}

\subparagraph{Theorem~\ref{thm:Parikh-image}.}

Let \((G,W)\) be a cycle-free WCFG and let \(S\) be the algebraic system in commuting variables corresponding to \((G,W)\).
Then, the strong solution \(r\) of \(S\) exists and the first component of \(r\) corresponds to \(Pk⟦G⟧_W\).

\begin{proof}
First, we prove that if a WCFG \((G,W)\) defined over a commutative and partially ordered semiring \(A\) is cycle-free, then the strong solution  of the algebraic system in commuting variables corresponding to \((G,W)\) exists.
Second, we show that the first component of the strong solution corresponds to \(Pk⟦G⟧_W\).

We give a proof of the first statement by contraposition by showing the following statement: let \(S\) be the algebraic system corresponding to a WCFG \((G,W)\) and let \(\sigma^0, \sigma^1, \ldots, \sigma^j, \ldots\) be the approximation sequence associated to \(S\). If \(\lim_{j\to\infty} \sigma^j\) does not exist, then \((G,W)\) is not cycle-free.
Note that \(\lim_{j\to\infty} \sigma^j\) does not exist iff  either the approximation sequence oscillates between a finite number of states, or there exists a length \(k\geq 0\) such that the coefficient of some monomial \(v \in Σ^{⊕} \), with \(|v| \leq k\), increases (w.r.t. the partial ordering of \(A\)) unboundedly at every step in the approximation sequence.
Formally, there exists \(k\geq 0\) such that, for every \(m \geq 0\), \(R_k(\sigma^{m}) \leq R_k(\sigma^{m+j})\), for some \(j> 0\), where \(\leq\) is the partial ordering of \(A\).
The first case cannot hold because every approximation sequence is monotonic~\cite[Lemma 14.4]{KuichSalomaa1986}.
Now we see that, if the second case holds, then necessarily the corresponding WCFG \((G,W)\) is not cycle-free.

To give some intuition, let us consider the following simple scenario.
Consider that the set of variables \(V\) of \(G\) contains only 1 variable, say \(X\), and assume that the limit of the approximation sequence of its corresponding algebraic system  does not exist.
Intuitively, for each \(j\geq 0\), the monomials occurring in the finite series \(\sigma^j\) in the approximation sequence of the corresponding system \(S\) correspond to the monomials that can be produced by \(G\) in at most \(j\) derivation steps, and the coefficient of each monomial corresponds to its weight if only derivations of at most \(j\) steps are considered.
By hypothesis, there exists a length \(k\geq 0\) such that the coefficient of some monomial \(v \in Σ^{⊕} \), with \(|v| \leq k\), increases unboundedly at every step in the approximation sequence.
It means that, for each \(m\)-step derivation  sequence \(\pi_m\) of \(G\) with \(m\geq 1\) generating \(w\in \Sigma^*\) with \(\lbag w \rbag = v\), there is another \(l\)-step derivation sequence \(\pi_l\) with \(l>m\) such that \(\pi_l\) generates \(w'\in \Sigma^*\) with \(\lbag w' \rbag = v\).
In other words, there exist derivation sequences in \(G\) of arbitrary length.
Because the number of rules of \(G\) is finite and so is the number of words \(w \in \Sigma^*\) such that \(\lbag w \rbag = v\), then either \(G\) contains a rule of the form \(X \rightarrow X\), or \(G\) contains a rule of the form \(X \rightarrow \gamma\) with \(\gamma \in \{X\}^+\) and \(|\gamma|>1\), and a rule of the form \(X \rightarrow \varepsilon\).
It follows that there exists a derivation sequence in \(G\) of the form \(X \Rightarrow^+ X\) and thus, \(G\) is not cycle-free.

The proof of the statement for every WCFG \((G,W)\) with an arbitrary number of variables goes in a similar fashion.
By hypothesis, there exists a length \(k\geq 0\) such that for some monomial \(v \in Σ^{⊕} \), with \(|v| \leq k\), for every \(m\)-step derivation  sequence \(\pi_m\) of \(G\) (\(m \geq 1\)) generating \(w\in \Sigma^*\) with \(\lbag w \rbag = v\), there is another \(l\)-step derivation sequence \(\pi_l\) with \(l>m\) such that \(\pi_l\) generates \(w'\in \Sigma^*\) with \(\lbag w' \rbag = v\).
That is, there are arbitrarily large derivation sequences in \(G\) using rules that do not add alphabet symbols.
Since the number of grammar rules of \(G\) is finite and so is the number of words \(w \in \Sigma^*\) such that \(\lbag w \rbag = v\), there must exist a cycle in \(G\), i.e., a derivation sequence of the form \(X_i \Rightarrow^+ X_i\), where \(X_i\) is a variable of \(G\).
Then, we conclude that \(G\) is not cycle-free.

We have shown that the strong solution \(r\) of \(S\) exists.
Now we prove that the first component of \(r\) corresponds to \(Pk⟦G⟧_W\).
First, consider \(\widetilde{S}\) as the algebraic system in \emph{noncommuting} variables corresponding to a cycle-free \((G,W)\) that is built as follows:
\begin{equation}
\label{eq:noncommuting_new}
X_i = \sum\limits_{\substack{\pi \in R \\ \pi = (X_i \rightarrow \gamma)}} W(\pi)\,\gamma \enspace .
\end{equation}
Note that \eqref{eq:noncommuting_new} now is of the form:
\begin{equation*}
X_i = p_i \enspace , \text{ with \(p_i \in A⟨(Σ \cup V)^{*}⟩\)} \enspace .
\end{equation*}
Salomaa et al. prove that, if \(\widetilde{r_1}\) is the first component of the strong solution of \(\widetilde{S}\), then \(\widetilde{r_1} (w) = ⟦G⟧_{W}(w) \) for every \(w \in  Σ^*\) when the weight function \(W\) is defined over \((\mathbb{N}, +, \cdot, 0, 1) \) and assigns \(1\) to each rule in \(G\)~\cite[Theorem 1.5]{Salomaa1978}.
In the proof they denote \(⟦G⟧_{W}(w)\) by \(amb(G,w)\) as it corresponds to the ambiguity of \(w\) according to \(G\).
The proof for the more general case where \(W\) is any arbitrary weight function defined over a commutative semiring reduces to replacing \(⟦G⟧_W(w)\) by \(amb(G,w)\) and using the corresponding semiring operations.

Now consider \(S\) as the algebraic system in commuting variables corresponding to \((G,W)\) and built as in \eqref{eq:algebraic-system} (page \pageref{eq:algebraic-system}).
Let \(r_1\) be the first component of its solution.
It is known that \(r_1\) and \(\widetilde{r_1}\) verify the following equality~\cite{KuichSalomaa1986}.
For each \(v \in Σ^{⊕}\):
\[r_1 (v) = \sum\limits_{\substack{v = \lbag w \rbag \\ w \in Σ^*} }\widetilde{r_1}(w) \enspace .\]
Then for each \(v \in Σ^{⊕} \):
\begin{equation*}
Pk⟦G⟧_W(v) = \sum\limits_{\substack{v = \lbag w \rbag \\ w \in Σ^*}} ⟦G⟧_W(w) = \sum\limits_{\substack{v = \lbag w \rbag \\ w \in Σ^*}} \widetilde{r_1}(w) = r_1(v) \enspace .
\end{equation*}

\end{proof}

\subparagraph{Lemma~\ref{lemma:regular-prop}.}

Let \((G,W)\) be a cycle-free WCFG.
Then \((G,W)\) satisfies the Parikh property iff \(Pk⟦G⟧_W \in A^{rat}⟨⟨Σ^{⊕}⟩⟩\).
\begin{proof}
Let \((G,W)\) be a WCFG with \(Pk⟦G⟧_W \in A^{rat}⟨⟨Σ^{⊕}⟩⟩\).
Then \(Pk⟦G⟧_W\) is the first component of the solution of a regular algebraic system \(S\).
Hence, the  WCFG \((G_{\ell}, W_{\ell})\) corresponding to \(S\) is regular and \(Pk⟦G⟧_W = Pk⟦G_{\ell}⟧_{W_{\ell}}\).

Now, let \((G,W)\) be a WCFG with the Parikh property.
Then there exists a regular WCFG \((G_{\ell}, W_{\ell})\) such that \(Pk⟦G⟧_W = Pk⟦G_{\ell}⟧_{{W_{\ell}}}\).
Let \(S\) be the regular algebraic system corresponding to \((G_{\ell}, {W_{\ell}})\).
The first component of its solution vector is \(Pk⟦G_{\ell}⟧_{W_{\ell}}\) and thus it is in \(A^{rat}⟨⟨Σ^{⊕}⟩⟩\).
As  \(Pk⟦G⟧_W = Pk⟦G_{\ell}⟧_{W_{\ell}}\) then \(Pk⟦G⟧_W\) is also in \(A^{rat}⟨⟨Σ^{⊕}⟩⟩\).
\end{proof}

\section{Unary polynomially ambiguous WCFGs are nonexpansive}
\label{sec:appendix-poly}

Bhattiprolu et al.~\cite{Vijay17} consider the class of \emph{polynomially ambiguous} WCFGs over the unary alphabet.
They show that every WCFG in this class satisfies the Parikh property.
Now we show that in the unary case the class of nonexpansive CFGs strictly contains the class of polynomially ambiguous grammars.

First, we introduce the definitions we will use in this section (some of them as given in~\cite{Vijay17}).
For convenience, we will adopt the standard way to write parse trees as labeled trees where nodes are either variables or terminals (see the definition of parse tree and yield of a parse tree in Chapter 5 of \cite{Ullman2003}).
Note that, for parse trees defined as in \cite{Ullman2003}, the yield may contain variables (as opposed to the yield of a parse tree as defined in Section~\ref{sec:preliminaries} which is always a (possibly empty) sequence alphabet symbols of \(\Sigma\)).
We denote by \(\mathcal{D}_G\) the set of parse trees of a grammar \(G\) when they are defined as in \cite{Ullman2003}.
We denote by \(\mathcal{D}_G(X)\) the set of all parse trees \(τ\) in \(\mathcal{D}_G\) of the form \(τ = X(\ldots)\), i.e., rooted at variable \(X\).
A parse tree \(τ ∈ \mathcal{D}_G(X)\) is said to be a \emph{\(X\)-pumping tree} if \(𝒴(τ)⇃_V = X\).
The set of all \(X\)-pumping trees is \(\mathcal{D}^{P}_{G}(X)≝ \{τ ∈ \mathcal{D}_G(X) \mid 𝒴(τ)⇃_V = X\}\).
The set of all pumping trees of \(G\) is given by \(\mathcal{D}^{P}_{G}≝\{τ ∈ \mathcal{D}_G(X) \mid X ∈ V\}\).
We define the concatenation of two parse trees \(\tau_1, \tau_2 \in \mathcal{D}_G\) with \( 𝒴(τ_1)⇃_V \neq \varepsilon\), denoted by \(\tau_1 \circ \tau_2\), by identifying the root of \(\tau_2\) with the first variable of \(\mathcal{Y}(\tau_1)\).
A set of trees \(T\) is \emph{ambiguous} if there are two distinct trees \(τ\) and \(τ'\) such that \(𝒴(τ) = 𝒴(τ')\).
Otherwise, \(T\) is \emph{unambiguous}.
They define the \emph{ambiguity function} of a CFG \(G\) as a function \(\mu_G: \mathbb{N} \mapsto \mathbb{N}\) such that \(\mu_G(n) ≝ max \{\len{parse_G(w)} : w \in \Sigma^*, \len{w} = n\}\), where \(parse_G(w)\) denotes the set of all derivations of \(G\) that generate \(w\).
A grammar is \emph{polynomially ambiguous} iff its ambiguity function \(\mu_G(n)\) is bounded by a polynomial \(p(n)\).
Finally, it is known that a CFG \(G\) is polynomially ambiguous iff \(\mathcal{D}^{P}_{G}\) is unambiguous~\cite{Wich1999}.

Now we show the main result of this section. 
\begin{theorem}
Every unary polynomially ambiguous CFG is nonexpansive.
\end{theorem}

\begin{proof}
The proof goes by contradiction.
Let \(G = (V, \{a\}, S, R)\) be a polynomially ambiguous grammar and assume that \(G\) is expansive.
Then, there is a derivation sequence of the form \(X \Rightarrow^{*} w_0\,X\,w_1\,X\,w_2\) with \(X ∈ V\) and \(w_i ∈ (V \cup Σ)^*\).
Assuming that every derivation sequence in \(G\) can produce a word of terminals (i.e., \(G\) does not contain useless rules), there exist necessarily at least two distinct parse trees \(\tau_1 = X(\ldots)\) and \(\tau_2 = X(\ldots)\) with \(\mathcal{Y}(\tau_1), \mathcal{Y}(\tau_2) \in \Sigma^*\) (not necessarily \(\mathcal{Y}(\tau_1)\neq \mathcal{Y}(\tau_2)\)).
Let \(\tau\) be the parse tree that corresponds to the derivation sequence \(X \Rightarrow^{*} w_0\,X\,w_1\,X\,w_2\) and consider the pumping trees \(\tau\circ \tau_1\) and \(\tau \circ \tau_2\) with
\[\mathcal{Y}(\tau \circ \tau_1) = w_0\,\mathcal{Y}(\tau_1)\,w_1\,X\,w_2\]
\[\mathcal{Y}(\tau \circ \tau_2) = w_0\,\mathcal{Y}(\tau_2)\,w_1\,X\,w_2 \enspace .\]
Now define the \(X\)-pumping trees \(\tau'\) and \(\tau''\) as follows:
\[\tau' = (\tau \circ \tau_1)\circ (\tau \circ \tau_2)\]
\[\tau'' = (\tau \circ \tau_2)\circ (\tau \circ \tau_1) \enspace .\]
Then,
\[\mathcal{Y}(\tau') = w_0\,\mathcal{Y}(\tau_1)\,w_1\,w_0\,\,\mathcal{Y}(\tau_2)\,w_1\,X\,w_2\,w_2\]
\[\mathcal{Y}(\tau'') = w_0\,\mathcal{Y}(\tau_2)\,w_1\,w_0\,\,\mathcal{Y}(\tau_1)\,w_1\,X\,w_2\,w_2 \enspace .\]
Because the alphabet is unary,
\[\mathcal{Y}(\tau') = \mathcal{Y}(\tau'') \enspace .\]
However, \(\tau'\ \neq \tau''\).
As there exist two distinct trees in \(\mathcal{D}^P_G\) with the same yield, \(\mathcal{D}^P_G\) is ambiguous (contradiction).
\end{proof}

We show that  the converse is not true with the following counterexample.
\begin{example}
Let \(G = (\{X,Y\}, \{a\}, X, \{X \rightarrow aXY,\,X \rightarrow aYX,\,X \rightarrow a,\,Y \rightarrow a\})\).
Note that \(X\) produces derivation sequences with at most one occurence of itself, and \(Y\) only produces one terminal symbol.
Thus, \(G\) is nonexpansive.
However, there are two distinct \(X\)-pumping trees \(\tau_1\) and \(\tau_2\) with \(\mathcal{Y}(\tau_1) = \mathcal{Y}(\tau_2)\) (Figures~\ref{fig:test1} and~\ref{fig:test2}).
Then, \(G\) is not polynomially ambiguous.

\begin{figure}[h]
\begin{minipage}{.5\textwidth}
  \centering
	\begin{forest}
	[X [a] [X [a] [Y [a]] [X]] [Y [a]]]
	\end{forest}
	\caption{\(τ_1 = X( a, X( a, Y(a), X),Y(a))\)}
  \label{fig:test1}
\end{minipage}%
\begin{minipage}{.5\textwidth}
  \centering
	\begin{forest}
	[X [a] [Y [a]] [X [a] [X] [Y [a]]]]
	\end{forest}
	\caption{\(τ_2 = X( a, Y( a), X( a, X, Y(a)))\)}
  \label{fig:test2}
\end{minipage}
\end{figure}
\qed
\end{example}

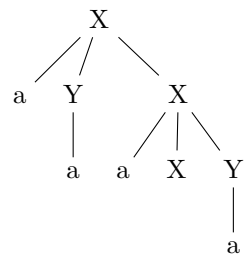
\end{document}